\documentclass[11pt,draftcls,onecolumn,journal]{IEEEtran}

\usepackage{latexsym}
\usepackage{verbatim}

\usepackage{cite}

\ifCLASSINFOpdf          
  \usepackage[pdftex]{graphicx}
  \DeclareGraphicsExtensions{.pdf,.jpeg,.png}
\else
  \usepackage[dvips]{graphicx}
  \DeclareGraphicsExtensions{.eps}
\fi

\usepackage[cmex10]{amsmath}
\usepackage{algorithmic}

\usepackage{array}

\usepackage{amsmath}
\usepackage{mdwmath}
\usepackage{mdwtab} 

\usepackage{eqparbox}

\usepackage[caption=false,font=footnotesize]{subfig}  

\usepackage{fixltx2e}

 \let\MYorigsubfloat\subfloat
 \renewcommand{\subfloat}[2][\relax]{\MYorigsubfloat[]{#2}}

\renewcommand {\vec}{\mathbf}

\begin{document}    

\title{Understanding of Multi-Frequency Interferometry Ranging: Analysis, Optimization and Validation }
\title{Optimization in Multi-Frequency Interferometry Ranging: Theory and Experiment }

\author{Li~Wei,~Wangdong~Qi,

\thanks{L. Wei is with the PLA University of Science and Technology, Nanjing 210007, China
        (e-mail: wlnb@hotmail.com).}

\thanks{Manuscript received ~x ~x, 2012; revised ~x ~x, 2012.}}





\markboth{submitted to IEEE Trans. ~x ~x,~Vol.~x, No.~x, x~2012}%
{Wei \MakeLowercase{\textit{et al.}}: Frequency Optimization in Multi-Frequency Interferometry Ranging} 




\maketitle


\begin{abstract}

Multi-frequency interferometry (MFI) is well known as an accurate phase-based measurement scheme. The paper reveals the inherent relationship of the unambiguous measurement range (UMR), the outlier probability, the MSE performance with the frequency pattern in MFI system, and then provides the corresponding criterion for choosing the frequency pattern. We point out that the theoretical rigorous UMR of MFI deduced in the literature is usually optimistic for practical application and derive a more practical expression . It is found that the least-square (LS) estimator of MFI has a distinguished ``double threshold effect". Distinct difference is observed for the MSE in moderate and high signal-to-noise ratio (SNR) region (denoted by MMSE and HMSE respectively) and the second threshold effect occurs during the rapid transition from MMSE to HMSE with increasing SNR. The closed-form expressions for the MMSE, HMSE and Cram\'{e}r-Rao bound (CRB) are further derived, with HMSE coinciding with CRB. Since the HMSE is insensitive to frequency pattern, we focus on MMSE minimization by proper frequency optimization. We show that a prime-based frequency interval can be exploited for the purpose of both outlier suppression and UMR extension and design a special optimal rearrangement for any set of frequency interval, in the sense of MMSE minimization. An extremely simple frequency design method is finally developed. Simulation and field experiment verified that the proposed scheme considerably outperforms the existing method in UMR as well as MSE performance, especially in the transition from MMSE to HMSE, for Gaussian and non-Gaussian channel.

\end{abstract}




\IEEEpeerreviewmaketitle

\section{Introduction}

  \IEEEPARstart{M}{ulti-frequency} interferometry (MFI) is an accurate phase-based ranging method, widely used in high-accuracy ranging\cite{Towers:03}\cite{Falaggis:09}, localization\cite{Kusy:06} or deformation and ground displacement detection\cite{Liu:11}, which estimates the range using the phase measurements recorded at multiple frequencies.

   How to select the measurement frequency in MFI system? There mainly exists three kinds of criterion.
   \begin{itemize}
     \item Maximizing the unambiguous measurement range (UMR).
     \item Minimizing the mean squared error (MSE) of range estimation.
     \item Decreasing the outlier (the wrong estimate lies outside the main lobe of the cost function\cite{Athley:05}) probability and enhancing the robustness to noise.
   \end{itemize}

   The paper will reveal the close relationship between these three criterions and the frequency pattern (refer to the frequency spacing of adjacent frequency in the paper) and pay special attention to the optimization of frequency pattern.

\subsection{Related work}

   Existing works have shown the UMR of some specially designed frequency patterns. Equally spaced frequency is employed in RIPS
   (Radio Interferometric Positioning System) \cite{Kusy:06} and the UMR is the synthetic wavelength of adjacent frequency, i.e. $c/\Delta f$. In \cite{Towers:03}, the geometric series of wavelengths is used and the UMR is extended to the synthetic wavelength of the closest two frequencies or the largest synthetic wavelength. The beat wavelength coincidence method is adopted in \cite{Falaggis:09} and the UMR is further enlarged by a factor of positive integer relative to that of \cite{Towers:03}.

   Meanwhile, ranging accuracy of multi-frequency interferometry system is another important target to be optimized. However, there exists limited works concerned with the accuracy improvement by frequency optimization as well as the performance analysis of ranging error in the literature.

   Among them, the estimation accuracy of RIPS method is in proportion to frequency separation, while the UMR is in inverse proportion to frequency separation\cite{Maroti:05}\cite{Kusy:06}. Then, there exists a compromise between the accuracy and UMR in RIPS.
   The multistage beat wavelength method is exploited for range estimation and frequency design in \cite{Towers:03} \cite{Falaggis:09}. The measurement phase noise will be amplified by $\Lambda_{0i}/\Lambda_{0,i+1}$ when two neighbor synthetic wavelengths $\Lambda_{0i}$,\;$\Lambda_{0,i+1}$ are used for phase unwrapping (determine the integer number of wavelength) in each stage. In order to maximize noise immunity, the wavelength ratio should be equal and the phase noise will be uniformly amplified in each stage. With this in mind, a geometric series of wavelengths is formed\cite{Towers:03}. However, the optimization criterion of \cite{Towers:03}\cite{Falaggis:09} is built on the beat wavelength method. Although it is a fast estimation method, it is far from the optimal one in term of estimation accuracy since only partial information is used in each stage.

   Note that similar problem arises in array design, that is, how to improve direction-of-arrival (DOA) estimation performance by optimizing the antenna positions of the $n$-element linear antenna array. The linear array geometry (antenna separation) design is most similar to our frequency pattern (frequency separation) design. So, we will also review the array geometry design of antenna array. The antenna positions are optimized to minimize the CRB of DOA estimation by numerical optimization in \cite{Gershman:97}. Based on the CRB principle, \cite{Doroslovacki:05} also considers the linear array optimization problem for joint range and DOA estimation.

   However, it seems not to be reasonable to optimize array geometry or frequency pattern only relying on CRB since it is a local measure of estimation accuracy.  The ambiguity function is introduced when the global accuracy performance and the outlier are evaluated   \cite{Doroslovacki:05}\cite{Haykin:91}\cite{Dogandzic:01}$\sim$\cite{Birinci:07}.

   Nonuniform linear array optimization is discussed in \cite{Haykin:91} with emphasis on the comparison between uniform array, nonredundant array and minimum redundant array. Array geometry is optimized in \cite{Gavish:96} to minimize the outlier probability. For each given geometry, the ambiguity function reaches its minimal at a certain direction and the optimal array is the one with the minimal value maximized. Both CRB and the ambiguity function are used to optimize the two-dimension array in \cite{Birinci:07}, and the genetic algorithm is adopted to search for the best array with the minimal CRB, under the constraint that the outlier probability is below a certain threshold.

   Outlier or sidelobe suppression is also an important topic in step frequency radar. Due to extremely large bandwidth, the estimation accuracy is easily satisfied and  more attention is paid to outlier suppression. To this end, random frequency step or different number of the frequency repetitions are proposed\cite{Axelsson:07}\cite{Rabideau:02}\cite{Levanon:02}. This is quite different from narrowband ranging. It is interesting to find in later section that the frequency pattern of our method is quite different from that in step frequency radar, where more weight is given to the frequencies distributed in the center instead of two ends of the band.

\subsection{Our work}

   It is known that the theoretical rigorous UMR of MFI is the least common multiple (LCM) of all wavelengths when expressed as
   integers by quantization\cite{McCormick:89}\cite{Wang:10}\cite{Vrana:93}. We prove that the UMR is the LCM of all synthetic
   wavelength of adjacent frequency or in inverse proportion to the greatest common divisor (GCD) of all frequency spacing when
   the initial frequency is properly chosen. Thus reducing the GCD will readily extend the measurement range. Compared with the
   wavelength-based method, the frequency spacing-based method can make full use of the band, in the sense that there have far more frequencies to be picked up for UMR extension within a given bandwidth\cite{Wang:11}. What is more, with the initial frequency fixed, it is possible to adjust the order of the frequency spacing (known as permutation) to generate various set of measurement
   frequencies, maintaining the same UMR, since the GCD of all frequency spacing remains unchanged. This forms the key foundation of later MSE optimization via the permutation of frequency spacing. It is also pointed out that the rigorous UMR is overly
   optimistic in case of random initial frequency and a more practical expression for the UMR is derived.

   The search based ambiguity function method is widely used to cope with the outlier in antenna array. However, it could not be directly applied to the MFI system because the obvious differences lie between the two systems. The parameter space of ambiguity function is ``bounded" to the direction of $-180^\circ\sim180^\circ$ for antenna array. It is not the case for MFI since the UMR will increase monotonously as the GCD of frequency spacing decreases.
   Moreover, constrained by the size and cost, the number of antenna is usually limited for array design while a large amount of frequencies may be required for accurate ranging in MFI.
   With this in mind, the ambiguity function of MFI is regarded as a theoretical foundation for the evaluation of a given frequency pattern or selecting the optimal frequency pattern from a very limited set.
   We also use the probability density function for analysis of the outlier probability and then design the frequency pattern with low outlier probability.

   The relationship between the theoretical attainable ranging accuracy and the frequency pattern is the major issue of the paper. It is well known that the LS estimator exhibits a threshold effect, i.e., the MSE increases sharply below a certain SNR.
   Surprising, it is found that the LS estimator of MFI system shows a distinguished ``double threshold effect", especially for
   $f_i\gg B$. Distinctive difference is observed for the MSE in moderate and high SNR region (denoted by MMSE and HMSE respectively). Beside the conventional threshold effect caused by outlier in low-SNR region, another threshold effect occurs
   during the rapid transition from large MMSE to small HMSE with increased SNR (see Fig.\ref{fig_double_mse}). The closed-form expression for the MMSE, HMSE and CRB are further obtained with HMSE coinciding with CRB.
   Under the assumption that $f_i\gg B$, which is true in most cases, different frequency design methods have almost the same HMSE
   and it is not the case for MMSE. Therefore, we focused on MMSE minimization by proper frequency optimization.

   A quite simple prime-based frequency design method is finally developed with no need of searching. That is, if priori knowledge is provided and the range uncertainty can be restricted to the mainlobe region of cost function, i.e., the outlier is excluded,
   the optimal measurement frequency is proved to be densely distributed on both ends of the band. Otherwise, a prime-based frequency interval is firstly constructed for the purpose of both sidelobe suppression and UMR extension, followed by a special optimal rearrangement of the frequency interval. The proof of the optimal rearrangement is given for the first time, in the sense of MMSE minimization.

   Different from the exhaust search or numerical optimization method, which have prohibitive computational complexity with increased antennas or measurement frequencies (Suppose $N$ measurement frequencies should be selected from $M$ possible frequencies, there exists $C_M^N$ combinations of measurement frequencies, so the antenna number $N=4$ is used in \cite{Gavish:96}), the proposed method may be free of the frequency number constraint.
   Moreover, for given bandwidth and frequency number, it is usually thought that the UMR and the ranging accuracy can not be improved simultaneously\cite{Cheng:11}. However, the proposed method and subsequent experiments have changed this opinion .

   Apart from the simulation analysis, field experiment in outdoor non-Gaussian channel (multipath error is incorporated inevitably) is also performed. The experiment results demonstrate that the proposed scheme considerably outperforms the existing method in UMR as well as MSE performance. Meanwhile, the simplicity and robustness even to non-Gaussian error are also attractive for its practical use.

   Notation: Upper (lower) bold face letters are used for matrices (column vectors).  $\vec{E}[x]$, $|x|$, $\angle{[x]}$, $\vec{Re}\{x\}$ and $\vec{Im}\{x\}$ denote the operation of taking expectation, absolute value, phase angle, real part and imaginary part of $x$, respectively. $[x]_{2\pi}$ denotes modulo-$2\pi$ operation, which reduces $x$ to the interval $(-\pi,\pi]$. $\hat{x}$ and $x^{\ast}$ denotes the estimate and optimal solution of $x$. $\vec{I}_N$ denotes the $N\times N$ identity matrix;  $\vec{1}_N$ denotes an $N\times1$ all-one column vector; $j=\sqrt{-1}$ ; The mth row and nth column entry of matrix $\vec{A}$ are denoted as $\vec{A}(m,n)$ or $\vec{[A]}_{m,n}$. The trace of $\vec{A}$ is given as $\mathrm{tr}(\vec{A})=\sum_m\vec{A}(m,n)$. $(\cdot)^T$, $(\cdot)^H$ and $(\cdot)^{-1}$ denote matrix transpose, conjugate transpose and inverse operators, respectively. We use script letters $\mathcal{A}$ to define sets and by $|\mathcal{A}|$ its cardinality.

\section{Problem Formulation and Performance Analysis}
    %

    \subsection{System model}

    In the absence of noise, the measurement phase of multi-frequency interferometry ranging system is related to the range $q_0$ by the following equation\cite{Maroti:05}\cite{Kusy:06}\cite{Towers:03}\cite{Falaggis:09}
    \begin{gather}
    \label{equa:basic}
        \varphi_0 (i)=\left[2\pi\frac{q_0}{\lambda_i} \right]_{2\pi}
    \end{gather}
    where $\lambda_i=c/f_i$ is the wavelength of the carrier frequency $f_i$, and $c$ is the speed of light, $q_0$ is the true range (it is the q-range in radio interferometry and path length difference in optical interferometry)and $\varphi_0 (i)$ is the ideal measurement phase wrapped to the principal interval of $(-\pi,\pi]$. The subscript ``0" denotes the true or ideal value.
    Equation (\ref{equa:basic}) is equivalent to
    \begin{gather}
        q_0 =\varphi_0 (i)\frac{\lambda_i }{2\pi} + n_i \lambda_i
    \end{gather}
    where $n_i$ is an unknown integer. For single measurement frequency, the range must satisfy $|q_0 |<\lambda_i/2$ and then $q_0 = \varphi (i)\lambda_i/{2\pi}$ with $n_i=0$. Otherwise, ambiguity appears because $n_i$ is undetermined. Therefore, multi-frequency ranging is usually needed.\\
    \indent There are various methods to estimate $q_0$ with a set of phase $\varphi (i)$ given, such as the CRT method \cite{Huang:87}\cite{Xia:07}\cite{Wang:10}, the LS-based search method \cite{Kusy:06}, the modulo conversion method\cite{Sundaram:00}, the multistage unwrapping method\cite{Towers:03}\cite{Cheng:85}\cite{Vrana:93}\cite{Falaggis:09}, the excess fractions method\cite{Benoit:1898}. Among which, the LS-based search method can achieve best estimation performance and will be adopted in the paper (we concentrate on the frequency design rather than estimation method), in the form of
    \begin{align}
    \label{equa:LS_estimation}
       &\hat q = \underset {q}{\textrm{argmin}}\; S(q),\quad  S(q)=\sum_{i=1}^N \left(\left[\varphi (i)-\hat \varphi_q (i)\right]_{2\pi}\right)^2 \notag \\
       &\hat \varphi_q (i) =\left[2\pi\frac{q}{\lambda_i} \right]_{2\pi},\quad
       \varphi(i)=\left[2\pi\frac{q_0}{\lambda_i} +\theta_e (i)\right]_{2\pi}
    \end{align}
    where $S(q)$ denotes the cost function and $\varphi(i)$ is the measurement phase with phase error $\theta_e (i)$.

    \subsection{Maximum Unambiguous Range}
    \label{subsect:UMR}
    Before a formal derivation of the unambiguous range, we begins with an intuitive interpretation of the ambiguity problem.\\
    \indent Fig. \ref{fig_intuition} shows the relationship between theoretical and measurement phase without noise under different
    range $q$. The theoretical and measurement phase are represented by red lines and blue stars. Although the true range $q$ as
    well as the theoretical phase are different, the measurement phase are the same with  $\Delta f=1.5\textrm{MHz}$. Note that
    the measurement phase is all the information we have in estimating $q$. In other words, the range $q=50$ and $q=250$ are
    indistinguishable for $\Delta f=1.5\textrm{MHz}$ and the ambiguity happens. If we measure the phase with a frequency
    spacing of $\Delta f=0.5\textrm{MHz}$ instead, the difference will immediately appear. This suggests that the ambiguity is
    closely related to the frequency spacing.
    \begin{figure}[!t]  
     \centerline{
              \subfloat[Case I]{\includegraphics[height=2.8in,width=3in]{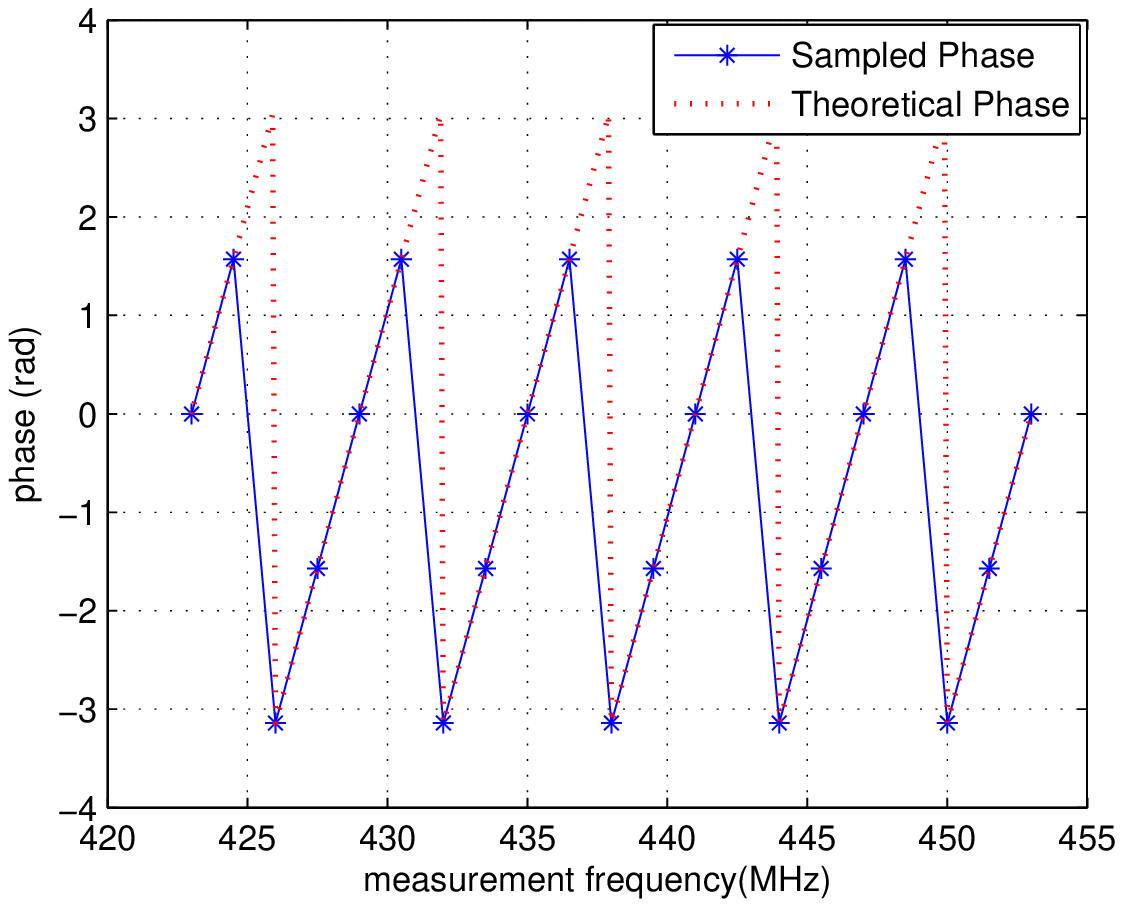}%
               \label{fig_intuition1}}
     \hspace{-0mm}
              \subfloat[Case II]{\includegraphics[height=2.8in,width=3in]{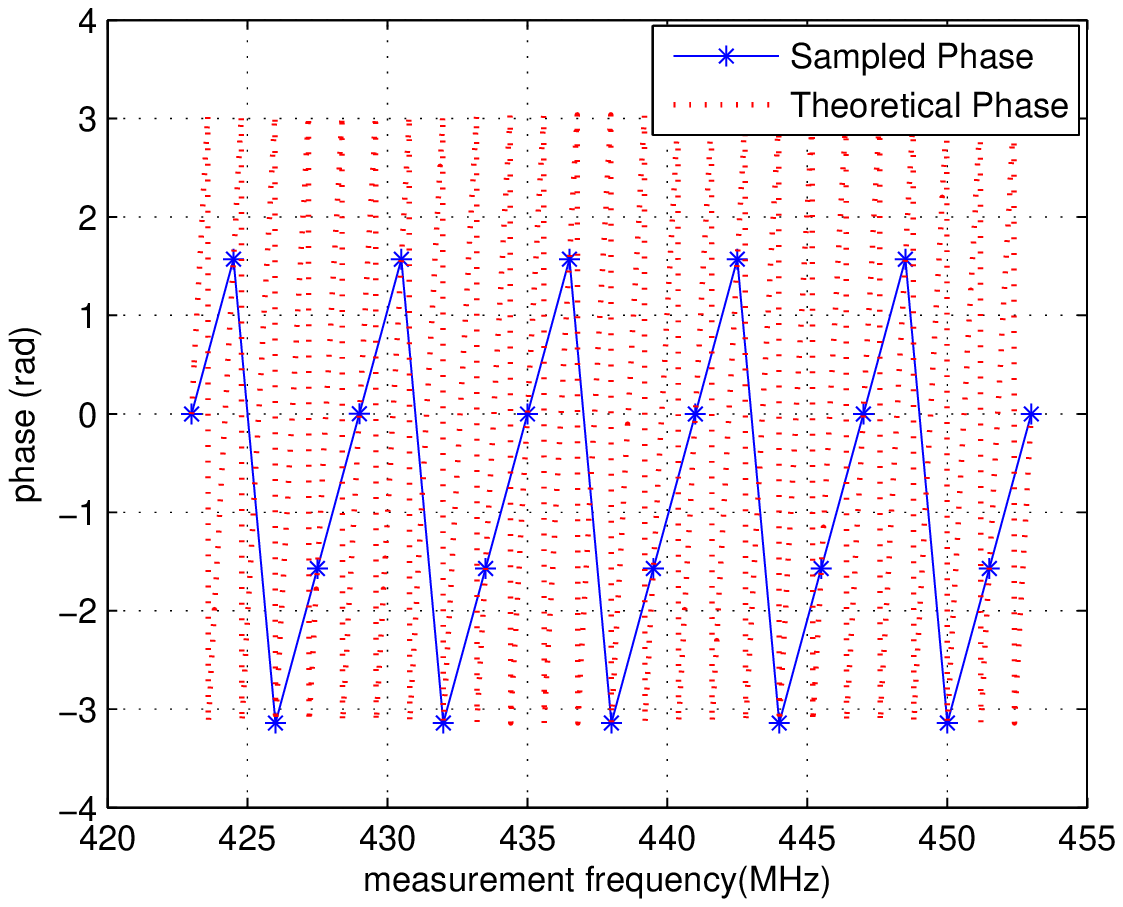}%
               \label{fig_intuition2}}
               }
     \caption{The sampled and theoretical wrapped phase versus frequency. (a) $q=50\,\textrm{m}$,$\Delta f=1.5\,\textrm{MHz}$.
               (b) $q=250\,\textrm{m}$,$\Delta f=1.5\,\textrm{MHz}$.}%
     \label{fig_intuition}
   \end{figure}

    For multi-frequency interferometry ranging, the phase under $q+\Delta L$ and $q$ are given by
    \begin{align}
        \hat \varphi_q (i)&=\left[ 2\pi f_i \frac{q}{c} \right]_{2\pi} , i = 1,2 \cdots N \notag\\
        \hat \varphi_{q+\Delta L} (i) &= \left[ {2\pi f_i \frac{q+\Delta L}{c}} \right]_{2\pi}
    \end{align}
    Ambiguity occurs if and only if the ranges $q$ and $q+\Delta L$ have the same phase in all measurement frequencies  $f_i, i=1,2,\cdots N$. That is,
    \begin{align}
    \label{equa:UMR_condition}
        \left[ ( \hat \varphi_{q+\Delta L}(i) - \hat \varphi_{q} (i) ) \right]_{2\pi}  &=0, \quad i = 1,2 \cdots N \notag\\
        \left[\frac{2\pi }{c}f_i \Delta L \right]_{2\pi} &= 0
    \end{align}
    Then
    \begin{gather}
    \label{equa:UMR_CRT}
        \Delta L  = n_1 \lambda_1 =n_2 \lambda_2 = \cdots n_N \lambda_N
    \end{gather}

    Where $n_i \in \mathbf{Z}$ and  $\lambda_i$ is the wavelength of frequency $f_i$.
    Clearly, there are indefinitely many solutions of $\Delta L$ satisfying equation system in (\ref{equa:UMR_condition}). Obviously, the minimum among those solutions is defined as the unambiguous measurement range (UMR). Equation (\ref{equa:UMR_CRT}) implies that the UMR is the least common multiple of all wavelengths\cite{Vrana:93}.

    To get closed-form solutions of $\Delta L$ in terms of frequency spacing instead of wavelength, equation system in (\ref{equa:UMR_condition}) is transformed into an equivalent form as follows:
    \begin{itemize}
      \item Condition 1: Equal initial phase at initial frequency $\hat \varphi_{q+\Delta L} (1)=\hat \varphi_q (1)$.
      \item Condition 2: Equal phase increment between adjacent frequencies $\Delta \hat\varphi_q(i,i+1)=\Delta\hat\varphi_{q+ \Delta L} (i,i+1)$.
    \end{itemize}
    where $\Delta \hat \varphi_q (i,j) = \left[\hat \varphi_q (j)-\hat \varphi_q (i) \right]_{2\pi} $.

    From Condition 1, we get
    \begin{align}
       \left[2\pi f_1 \frac{q}{c}\right]_{2\pi}&= \left[2\pi f_1 \frac{q+\Delta L}{c} \right]_{2\pi} \notag\\
       \Delta L&=k_0 c/{f_1} = k_0 \lambda _1
    \end{align}

    From Condition 2, we get
    \begin{align}
        \left[2\pi \Delta f_i \frac{q}{c} \right]_{2\pi}&=\left[2\pi \Delta f_i \frac{q + \Delta L}{c} \right]_{2\pi}  \notag\\
        \Delta L &= k_i c/{\Delta f}_i
    \end{align}

      where $\Delta f_i= f_{i+1}-f_i$ (assume $f_1<f_2<\cdots<f_N$) and $k_0,k_i$ are positive integer.
       \newtheorem {theorem}{\noindent Theorem}
       \begin{theorem}
        \label{theo:UMR}
         If $\Delta f_{\min} $ is the greatest common divisor (GCD) of all the adjacent frequency spacing and $f_1=k\Delta f_{\min} $, then the UMR is  $\Delta L=c/\Delta f_{\min}$.
       \end{theorem}
       \begin{proof}
         Without loss of generality, we assume $f_{i+1}>f_i$ and the separation of two adjacent frequencies is $f_{i + 1}-f_i=
         \Delta f_i =k_i\Delta f_{\min}$, $k_i$ is a positive integer. It is clear that $\Delta L=c/\Delta f_{\min}=k_ic/\Delta f_i$, and $\Delta L=kc/f_1$, then $\Delta L=c/\Delta f_{\min}$  is one of the ambiguous range. If it is not the minimum ambiguous range (UMR), let $\Delta L'=\Delta L/m=c/(m\Delta f_{\min})$ be the minimum ambiguous range, $m$ is positive integer, using condition 2, then
         \begin{align}
         \label{equa:mini_UMR}
         \frac{{2\pi }}{c}k_i \Delta f_{\min } \Delta L' = p_i 2\pi \notag\\
           f_{i+1}-f_i=mp_i\Delta f_{\min}
         \end{align}
        where $p_i$ is arbitrary positive integer. (\ref{equa:mini_UMR}) indicates that the greatest common divisor is $m\Delta f_{\min}$ instead of $\Delta f_{\min}$. This conclusion conflicts with the hypothesis. Therefore, we claim that $c/\Delta f_{\min}$ is the UMR.
       \end{proof}

       \newtheorem {corollary}{\noindent Corollary}
       \begin{corollary}
       \label{coro:RIPS}
       For equal-spaced measurement frequency, i.e., $\Delta f_i=\Delta f$, and $f_1=k\Delta f$, then the UMR is $\Delta L=c/\Delta f$.
       \end{corollary}

       \begin{theorem}
       \label{theo:UMR_modify}
        Suppose $\Delta f_{\min} $ is the GCD of all the adjacent frequency spacing, if the condition $f_1=k\Delta f_{\min} $ is not met in Theorem \ref{theo:UMR} and let $f_1 =(k_1 +\varepsilon)\Delta f_{\min}$, $-0.5<\varepsilon\le0.5$, $q_0$ is the true range, then
        \renewcommand{\labelenumi}{(\arabic{enumi})}
        \begin{enumerate}
        \renewcommand{\baselinestretch}{0.2}\normalsize
         \item	The cost function $S(q)$ will achieve a local minimum at $q=q_0 +c/\Delta f_{\min}- \varepsilon{\sum_{i=1}^N \lambda_i ^{-1}}\bigg/{\sum_{i=1}^N \lambda_i^{-2}}$, with
           \begin{gather}
              \hspace{-30pt}\underset{f_1/B\rightarrow\infty}{\lim}
                 S(q|q=q_0 +c/\Delta f_{\min}- \varepsilon{\sum_{i=1}^N \lambda_i ^{-1}}\bigg/{\sum_{i=1}^N \lambda_i^{-2}})=0
           \end{gather}
           where $S(q)$ is defined in (\ref{equa:LS_estimation}).
         \item	 For any $f_1/B$, there does not exist any $\Delta l$ with $0<\Delta l<c/\Delta f_{\min}- \varepsilon{\sum_{i=1}^N \lambda_i ^{-1}}\bigg/{\sum_{i=1}^N \lambda_i^{-2}}$, satisfy
           \begin{gather}
              \hspace{-30pt}\underset{f_1/B\rightarrow\infty}{\lim}
                 S(q|q=q_0 +\Delta l)=0
           \end{gather}
         \item  Assume that the measurement phase noise $\theta_e (i)$ at each frequency $f_i$ is independent and identically distributed zero-mean Gaussian noise with variance $E\{\theta_{e}^{2}(i)\}=\sigma_{\theta}^{2}$. If $f_1/B\geq4$ and $\textrm{SNR}>0$(SNR is defined in section \ref{subsect:CRB}), then, with at least $3\sigma$ reliability(or $99.7\%$ reliability), we have
             \begin{align}
                &P_a=P\left( S\left(q_0+c/\Delta f_{\min}\hspace{-2pt}-\hspace{-2pt}\varepsilon{\sum_{i=1}^N \lambda_i^{-1}}\bigg/{\sum_{i=1}^N \lambda_i^{-2}}\right)< S\left(q_0\right) \right)   \notag\\
                &>\frac{1}{2} \left(1- \textrm{erf}\left(\frac{\sqrt{N} W}{2\sqrt{2}\sigma_\theta } \right) \right)
             \end{align}
             where $\textrm{erf}\left(x\right)=\frac{2}{\sqrt{\pi}} \int_0^x {e^{-t^2}}dt$, and $W=2\pi|\varepsilon|\frac{B}{f_1}$.
         \vspace{-5pt}
       \end{enumerate}
       \end{theorem}

       \begin{proof}
          See Appendix \ref{append:UMR}.
       \end{proof}
   Remark:
   This implies that $c/\Delta f_{\min}-\varepsilon\frac{\sum_{i=1}^N \lambda _i^{-1}}{\sum_{i=1}^N \lambda_i^{-2} }$ is the minimal unambiguous range (UMR) only in the sense of limit.
   However, if $f_1/B\gg1$, the cost function at  $q_0 +c/\Delta f_{\min}- \varepsilon{\sum_{i=1}^N \lambda _i ^{-1}}\big/{\sum_{i=1}^N \lambda_i^{-2} }$ will be indistinguishable from the one at $q_0$, with an extreme large probability, even in the presence of extremely small noise. For example, when $f_1/B=10$,  $\varepsilon=0.1$, $N=40$, we have $P_a>30.8\%$ and $P_a>18.7\%$ for $\textrm{SNR}=5\textrm{dB}$ and $\textrm{SNR}=10\textrm{dB}$ respectively.
   We then define $\Delta L=c/\Delta f_{\min}- \varepsilon{\sum_{i=1}^N \lambda _i ^{-1}}\big/{\sum_{i=1}^N \lambda_i^{-2} }$ as the practical UMR, denoted as P-UMR. If the measurement range is beyond the limit of P-UMR, a very large error probability will appear, as is verified in Fig.\ref{fig_UMR_bandwidth}(a) and Fig.\ref{fig_UMR_experiment}. The UMR derived in Theorem \ref{theo:UMR} is a special case of Theorem \nobreak\ref{theo:UMR_modify} with $\varepsilon=0$.\\
   \indent Note that the UMR of  CRT method is the product of all wavelengths according to (\ref{equa:UMR_CRT}) when expressed as prime number by quantization\cite{McCormick:89}\cite{Huang:87} and is much larger than the P-UMR in Theorem \ref{theo:UMR_modify}, it becomes an overly optimistic estimation since it is not attainable in practice.\\
   \indent  Since $\varepsilon{\sum_{i=1}^N \lambda _i ^{-1}}\big/{\sum_{i=1}^N \lambda_i^{-2} }<\lambda_1=c/f_1$ is negligible compared to $c/\Delta f_{\min}$ under the assumption $f_1\gg B\gg \Delta f_{\min}$, the UMR is thought of $c/\Delta f_{\min}$ for both $f_1=k\Delta f_{\min} $ and $f_1\neq k\Delta f_{\min} $ in latter part of the paper.

   \subsection{Outlier}

   Inspired by the result in \cite{Gavish:96}\cite{Birinci:07}, we introduce the ambiguity function to cope with the outlier (the wrong estimate lies outside the main lobe of the cost function) in multi-frequency interferometry ranging.

   It is proved in \cite{Athley:05} \cite{Birinci:07} that the outlier probability is in proportion to the ambiguity function
   \begin{gather}
    \mathrm{C} (\Delta q)=\frac{ \left|  \mathbf{s}(q)^H  \mathbf{s}(q+\Delta q )  \right|^2 } {N^2}
      = \frac{1}{N^2}\left| \sum_{i=1}^N e^{j\frac{2\pi}{c}f_i \Delta q}  \right|^2
   \end{gather}
   where
   \begin{gather*}
      \mathbf{s}(q)=\left[e^{j2\pi f_1 q/c},e^{j2\pi f_2 q/c},\cdots e^{j2\pi f_N q/c}  \right]^T
   \end{gather*}


   It is obvious that the ambiguity function attains its global maximal $1$ when $\Delta q=k\Delta L$, where $k$ is an integer. This kind of ambiguity (UMR) is inherent and inevitable. The UMR is excluded with the constraint $|q_0|<\Delta L$. Meanwhile, for $\frac{B_m}{2} < | \Delta q | < \Delta L- \frac{B_m}{2}$, $B_m$ is the width of mainpeak, there still exists the second maximal which may exceed the value achieved at the true $q_0$ with the help of noise and lead to large error, i.e., outlier. The objective is to solve the following optimization problem

   \begin{gather}
    \underset{\Delta \mathbf{f} }{\min} \left\{  \underset {\Delta q} {\max}
    \frac{1}{N^2}\left| \sum_{i=1}^N e^{j\frac{2\pi}{c}f_i \Delta q}  \right|^2 \right\}\notag\\
    \mathrm{subject\; to} \quad \frac{B_m}{2} < | \Delta q | < \Delta L- \frac{B_m}{2}
   \end{gather}
   Since
   \begin{gather}
    \mathrm{C} (\Delta q ) = \frac{1}{N} +\frac{2}{N^2}\mathbf{Re} \left\{\sum_{i=1}^{N-1} \sum_{k=i+1}^N e^{ j\frac{2\pi (f_i-f_k)\Delta q}{c} } \right\}
   \end{gather}
   Then the problem becomes
   \begin{gather}
     \underset{\Delta \mathbf{f} }{\min} \left\{  \underset {\Delta q} {\max}
      \mathbf{Re}  \left\{   \sum_{i=1}^{N-1} \sum_{k=i+1}^N e^{ j\frac{2\pi (f_i-f_k)\Delta q}{c} }   \right\}   \right\}
   \end{gather}
   Note that we are interested in those $\Delta q$ satisfying $\left[\frac{2\pi}{c}(f_i-f_k)\Delta q \right]_{2\pi} \ll 1$ , which are prone to resulting in outlier. Using a first-order Taylor series expansion of $e^x$, we have

  \begin{align}
      \label{equa:Taylor}
      \mathbf{Re}  \left\{   \sum_{i=1}^{N-1} \sum_{k=i+1}^N e^{ j\frac{2\pi (f_i-f_k)\Delta q}{c} }   \right\}
      \simeq \sum_{i=1}^{N-1}\sum_{k=i+1}^N \left(1-\frac{1}{2}\left(\left[\frac{2\pi}{c}(f_i-f_k)\Delta q \right]_{2\pi}\right)^2  \right)
  \end{align}
  According to {\cite{Wu:95}}
  \begin{align}
      \sum_{i=1}^{N-1}\sum_{k=i+1}^N \left(\left[\frac{2\pi}{c}(f_i-f_k)\Delta q \right]_{2\pi}\right)^2
      = \left[\frac{2\pi}{c}\Delta q\Delta \mathbf{f}^T \mathbf{\Gamma}^T \right]_{2\pi} \mathbf{R}^{-1}
         \left[\frac{2\pi}{c}\Delta q \mathbf{\Gamma}\Delta \mathbf{f} \right]_{2\pi}
  \end{align}
  where
  \begin{gather*}
    \Delta \mathbf{f} =[\Delta f_1 ,\Delta f_2 , \cdots\Delta f_{N-1} ]^T,\quad
     \mathbf{\Gamma}= \left[\begin{array}{*{20}c}
                       1 & 0 & \cdots & 0  \\
                       1 & 1 & \cdots & 0  \\
                     \vdots  &  \vdots  & 1 & 0  \\
                       1 & 1 & 1 & 1  \\
                     \end{array}  \right]
   \end{gather*}
   \begin{gather}
   \label{equa:R_u}
      \vec{R}^{-1}\hspace{-2pt}=\hspace{-2pt}\vec{I}\hspace{-1pt}_{N\hspace{-1pt}-\hspace{-1pt}1}\hspace{-2pt}-\hspace{-2pt}\frac{\vec{uu}^T}{N},\, \vec{R}\hspace{-2pt}=\hspace{-2pt}\vec{I}\hspace{-1pt}_{N\hspace{-1pt}-\hspace{-1pt}1}\hspace{-2pt}+\hspace{-2pt}\vec{uu}^T,\,
      \vec{u}\hspace{-2pt}=\hspace{-2pt}[1,1,\cdots 1]^T \hspace{-2pt}= \hspace{-2pt}\mathbf{1}\hspace{-1pt}_{N\hspace{-1pt}-\hspace{-1pt}1}
   \end{gather}

   It becomes the max-mini optimization problem
   \begin{gather}
       \underset{\Delta \mathbf{f} }{\max} \left\{  \underset {\Delta q} {\min}
       \left[\frac{2\pi}{c}\Delta q\Delta \mathbf{f}^T \mathbf{\Gamma}^T \right]_{2\pi} \mathbf{R}^{-1}
       \left[\frac{2\pi}{c}\Delta q \mathbf{\Gamma}\Delta \mathbf{f} \right]_{2\pi} \right\}
   \end{gather}

   When $N$ is large, it is approximated as
   \begin{gather}
    \label{equa:cost_outlier}
      \underset{\Delta \mathbf{f} }{\max} \left\{  \underset {\Delta q} {\min}
      \left[\frac{2\pi}{c}\Delta q\Delta \mathbf{f}^T \mathbf{\Gamma}^T \right]_{2\pi}
      \left[\frac{2\pi}{c}\Delta q \mathbf{\Gamma}\Delta \mathbf{f} \right]_{2\pi} \right\}
   \end{gather}

   Note that (\ref{equa:cost_outlier}) may be used to judge the quality of a specific $\Delta \mathbf{f}$ as well as to search for an optimal $\Delta \mathbf{f}$, in term of minimal outlier probability. As is the usual case in array design \cite{Gavish:96} \cite{Birinci:07}.
   However, we are not intended to find the optimal $\Delta \mathbf{f}$ by exhaust searching due to the formidable complexity when the parameter space is too large.
   It is known that the probability density function (PDF) of $q$ given $q_0$ and the measurement phase variance  $\sigma_\theta^2$ at frequency $f_i$ is  \cite{Vrana:93}\cite{Pascazio:01}\cite{Cheng:11}
      \begin{gather}
        \label{equa:pdf_single}
        p(f_i|q)\approx \frac{1}{\sqrt {2\pi} \sigma_q(i) }   \exp \big[-\frac{1}{2\sigma_q^2(i)} ([q- q_0]_{\lambda_i})^2  \big] \notag\\
         \sigma_q(i)=\frac{c}{2\pi f_i} \sigma_\theta
      \end{gather}

     The joint probability density functions for measurement using two adjacent frequencies $f_i$ and $f_{i+1}$ becomes
       \begin{gather}
        \label{equa:pdf_double}
        p(f_i,f_{i+1}|q)= p(f_i|q)p(f_{i+1}|q)
        \approx\frac{1}{\prod_{k=i}^{i+1} \sqrt {2\pi} \sigma_q(k) }   \exp \Big[- \sum_{k=i}^{i+1} \frac{1}{2\sigma_q^2(k)} ([q- q_0]_{\lambda_k})^2  \Big]
      \end{gather}
     When $f_1\gg B$, $\sigma_q(k)\approx \sigma_q(0)$, denoted as $\sigma_q$, (\ref{equa:pdf_double}) can be simplified as
       \begin{gather}
        \label{equa:pdf_double}
        p(f_i,f_{i+1}|q)
        \approx\frac{1}{{2\pi} \sigma_q^2 } \exp \Big[-\frac{1}{2\sigma_q^2} \sum_{k=i}^{i+1}([q- q_0]_{\lambda_k})^2  \Big]
      \end{gather}

     \begin{figure}[!t]
      \centerline{
         \subfloat[Case I]{\includegraphics[height=2.8in,width=3in]{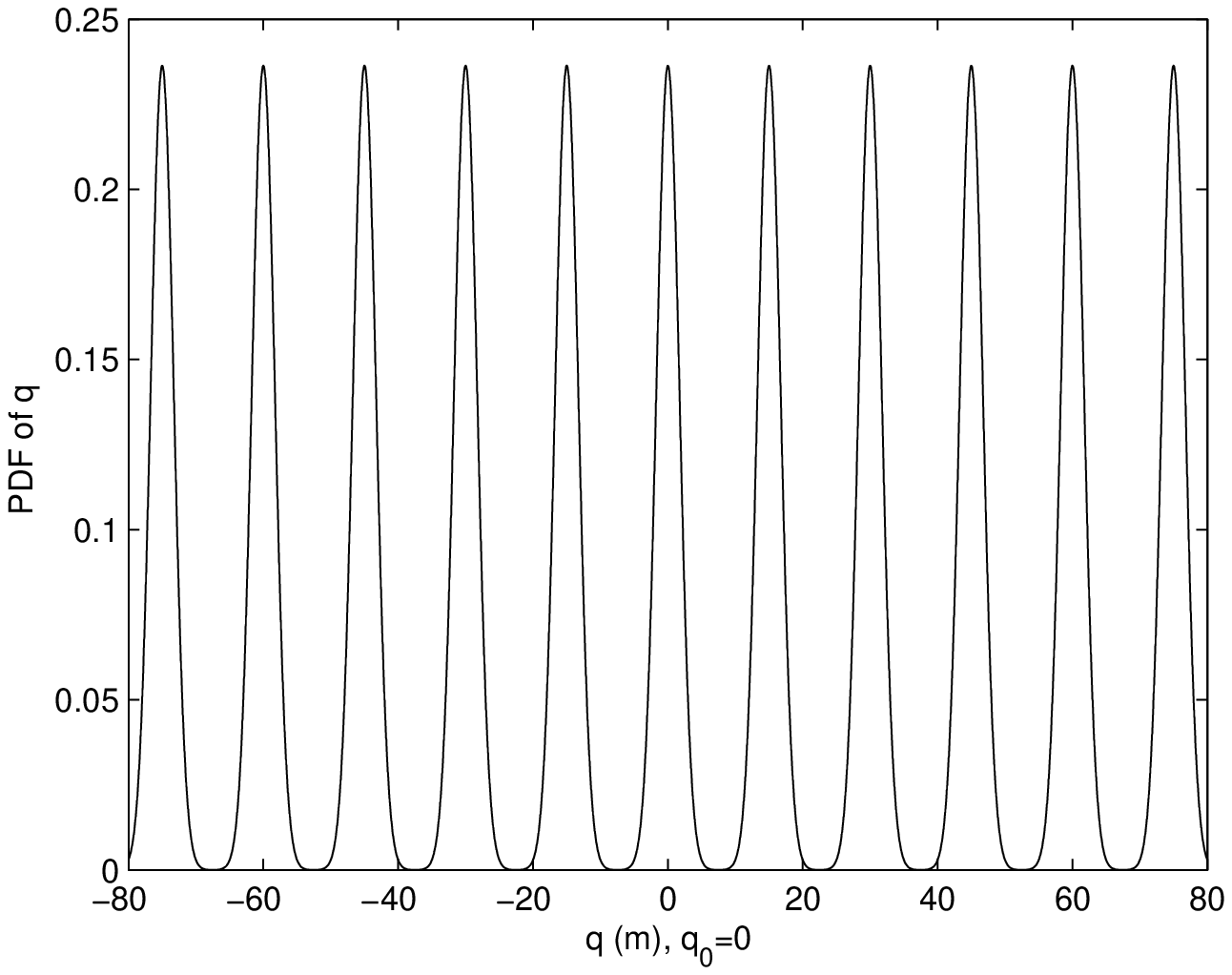}%
         \label{fig_PDF_f1}}
         \subfloat[Case II]{\includegraphics[height=2.8in,width=3in]{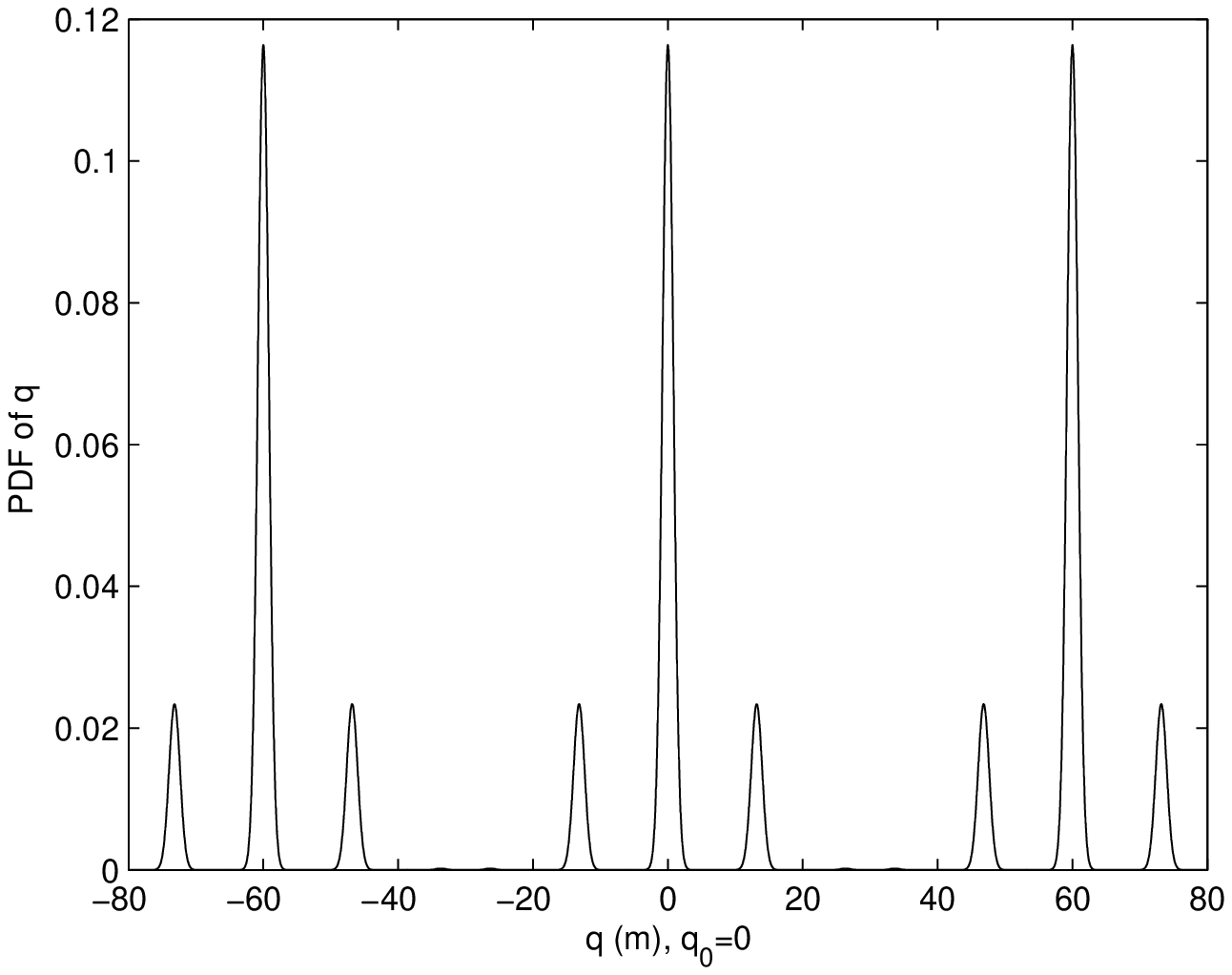}%
         \label{fig_PDF_f12}}
         }
         \hfill               
      \caption{The PDF of q with $\sigma_\theta^2=0.3$ and $q_0=0$. (a) Single frequency PDF $p(f_1|q)$ with $f_1=20\textrm{MHz}$. (b) Double-frequency PDF $p(f_1,f_2|q)$ with $f_1=20\textrm{MHz}$ and $f_2=25\textrm{MHz}$.}
      \label{fig_PDF_q}
    \end{figure}

     Let $\Delta q=q-q_0$,  it is proved in section \ref{subsect:UMR} that the UMR is well approximated by $c/\Delta f_{i}$ for two adjacent frequencies $f_i$ , $f_{i+1}$ satisfying $f_i\gg (f_{i+1}-f_i)$. So, $p(f_i,f_{i+1}|q)$ is $c/\Delta f_{i}$-periodic functions ( $p(f_i,f_{i+1}|q)=p(f_i,f_{i+1}|q+kc/\Delta f_{i})$, $k \in \mathbf{Z}$, see Fig.\ref{fig_PDF_q}). Then, sharp peak will appear at $\Delta q^{(i)}(k)=kc/\Delta f_{i}$.\\
     \indent Denote $\mathbf{f}=[f_1,f_2,\cdots,f_N]$, the PDF for $N$ independent measurement frequencies is
         \begin{gather}
          \label{equa:pdf_Multi}
            p(\mathbf{f}|q)=\prod_{k=1}^{N}p(f_k|q)=\sqrt { {\prod_{k=1}^{N-1}p(f_k,f_{k+1}|q)}* {p(f_1,f_{N}|q)} }
         \end{gather}
     (\ref{equa:pdf_Multi}) will be exploited in the design of $\Delta \mathbf{f}$ in later section.

    \subsection{MSE}

    Suppose the measurement phase at the $i\text{-th}$ frequency is $\varphi(i)=\left[2\pi q_0/\lambda_i +\theta_e (i) \right]_{2\pi}$, and the estimated phase is $\hat \varphi_q (i)=\nobreak\left[2\pi q/\lambda_i  \right]_{2\pi}$. where $q_0$, $q$ are the true and estimated range, $\theta_e (i)$ is assumed to be i.i.d white Gaussian noise. The LS-based objective function is given by
    \begin{gather}
        \label{equa:LS_object}
        \underset{q}{\min}  \sum_{i=1}^N  \left(\left[\varphi(i) - \hat \varphi_q(i) \right]_{2\pi} \right)^2
    \end{gather}

    Since $\left[[x]_{2\pi}-[y]_{2\pi}\right]_{2\pi}=[x-y]_{2\pi}$, then (\ref{equa:LS_object}) becomes
    \begin{gather}
        \label{equa:cost_mse}
        \underset{\Delta q}{\min}  \sum_{i=1}^N  \left(\left[\frac{2\pi f_i \Delta q}{c}-\theta_e(i) \right]_{2\pi} \right)^2
    \end{gather}
    We will derive the MSE using the perturbation analysis approach\cite{Chan:94}. The MSE for the moderate and high SNR are denoted by MMSE and HMSE in the paper. For SNR high enough, i.e., $\left|\Delta q \right|<\lambda_N\cdots <\lambda_1$. In this case, (\ref{equa:cost_mse}) may be simplified as
    \begin{gather}
        \underset{\Delta q}{\min} \sum_{i=1}^N \left(\frac{2\pi f_i \Delta q}{c}-\theta_e(i) \right)^2
    \end{gather}
    The optimal $\Delta q$ satisfy
    \begin{gather}
        \sum_{i=1}^N \left(\frac{2\pi f_i \Delta q}{c}-\theta_e (i) \right) f_i= 0 \notag\\
        \Delta q =\frac{c\sum_{k=1}^N f_k \theta_e (k) } {2\pi \sum_{k=1}^N {f_k^2} }
    \end{gather}
    Then the HMSE is readily obtained
    \begin{gather}
    \label{equa:HMSE}
        \text{HMSE}=\mathbf{E}\left[\Delta q^2 \right]=\frac{c^2 \sigma_\theta ^2 }{4\pi^2 }\left(\sum_{k=1}^N f_k^2 \right)^{-1}
    \end{gather}

     For moderate or high SNR, $\left|\Delta q \right|>\lambda_N$, but $\left|\Delta q\right|< c/2B$, then the error could not be
     obtained directly. Similar to (\ref{equa:Taylor}), we may transform the estimator into (since $\Delta q$ is the one
     minimizing (\ref{equa:cost_mse}), then it leads to small $({2\pi f_i \Delta q}/{c}-\theta_e(i))$ for moderate or high SNR):
    \begin{gather}
        \underset{\Delta q}{\max} \mathbf{Re}\left\{\sum_{i=1}^N \exp \left \{{j\left(\frac{2\pi f_i \Delta q}{c}-\theta_e(i) \right)} \right \} \right\}
    \end{gather}

    When $f_1/B\gg 1$, the problem may be well approximated as (the approximation is reasonable, as is confirmed in later simulation, see Fig.\ref{fig_comp_cost} and Fig.\ref{fig_double_mse})
    \begin{figure}[!t]
      \centering
      \includegraphics[height=2.8in,width=3.5in]{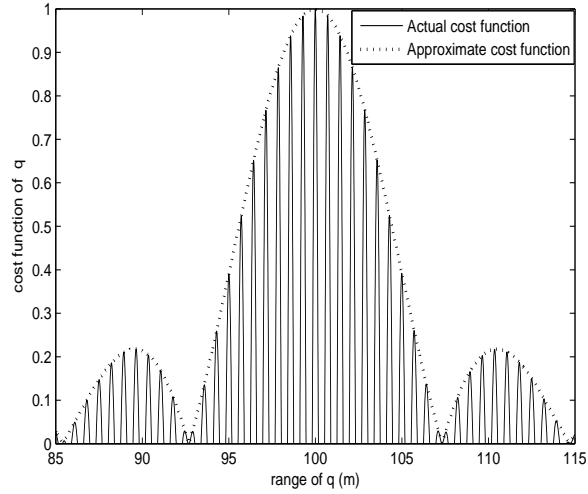}
      \caption{Comparison of the actual and approximate cost function  }%
      \label{fig_comp_cost}
    \end{figure}
    \begin{gather}
        \label{equa:cost_min}
        \underset{\Delta q}{\max} \left|\sum_{i=1}^N \exp \left \{{j\left(\frac{2\pi f_i \Delta q}{c}-\theta_e(i) \right)} \right \}    \right|^2   \notag\\
        \underset{\Delta q}{\max} \vec{Re} \left\{ \sum_{i=1}^{N-1} \sum_{k=i+1}^N
        \hspace{-3pt}\exp \left\{ {j\left( \frac{ 2\pi \left(f_i-f_k \right)\Delta q }{c}- \Big(\theta_e(i)-\theta_e (k)\Big) \right)} \right\} \right\}
     \end{gather}
    or equivalently
    \begin{gather}
        \label{equa:cost_min}
        \underset{\Delta q}{\min}\hspace{-3pt} \sum_{i=1}^{N-1}\hspace{-3pt} \sum_{k=i+1}^N \hspace{-4pt}
        \left( \left[\frac{ 2\pi\left(f_i\hspace{-2pt}-\hspace{-2pt}f_k \right)\Delta q }{c}- \hspace{-2pt}\Big(\theta_e(i)\hspace{-2pt}-\hspace{-2pt}\theta_e(k)\Big) \right]_{2\pi} \right)^2
    \end{gather}
    When $\left|\Delta q\right|< c/2B$, then $\left|\Delta q \right| < \left| c/2\left(f_i-f_k\right) \right|$. Omitting the influence of $\theta_e(i)-\theta_e(k)$ for moderate or high SNR, (\ref{equa:cost_min}) can be rewritten as
    \begin{gather}
        \underset{\Delta q}{\min} \sum_{i=1}^{N-1} \sum_{k=i+1}^N  \left( \frac{ 2\pi\left(f_i-f_k \right)\Delta q }{c}- \Big(\theta_e(i)-\theta_e(k)\Big)                          \right)^2
    \end{gather}
    Define
    \begin{gather*}
        \phi_{\Delta q} \left(i,k\right)=\frac{2\pi \left(f_i -f_k \right)\Delta q}{c} -\left(\theta_e(i)-\theta_e(k)\right)\\
        \vec{\Phi}(\Delta q)=\left[\phi_{\Delta q} (2,1),\phi_{\Delta q} (3,1),\cdots \phi_{\Delta q} (N,1) \right]^T \\
        \vec{\Theta}_e\hspace{-4pt}=\left[(\theta_e(2)-\theta_e(1)),(\theta_e(3)-\theta_e(1)),\cdots (\theta_e(N)-\theta_e(1)) \right]^T
    \end{gather*}
    The problem can be expressed as
    \begin{gather}
       \underset{ \Delta q}{\min} \sum_{i=1}^{N-1} \sum_{k=i+1}^N  \left(\phi_{\Delta q} (i,k) \right)^2
                           = \underset{\Delta q}{\min}  \vec{\Phi}^T (\Delta q)\vec{R}^{-1} \vec{\Phi}(\Delta q)
    \end{gather}
    Where $\vec{R}^{-1}$, $\vec{R}$, $\vec{u}$ are defined in (\ref{equa:R_u}).
    The optimal solution obey
    \begin{gather}
       \frac{ \partial \left(  \vec{\Phi}^T(\Delta q)  \vec{R}^{-1}   \vec{\Phi}(\Delta q)  \right)  }  {\partial (\Delta q)} = 0 \notag\\
       \left( \partial \left[\vec{\Phi}(\Delta q) \right]  {\partial (\Delta q)}  \right)^T  \vec{R}^{-1}  \vec{\Phi}(\Delta q)= 0
    \end{gather}

   Since $\partial \left[\phi_{\Delta q}(i,1) \right] / \partial(\Delta q)=\nobreak \frac{2\pi}{c}\left(f_i-f_1 \right) =\nobreak \frac{2\pi}{c}\sum_{k=1}^{i-1} {\Delta f_k}$, where $\Delta f_k=f_{k+1}-f_k $, and
    \begin{align}
        \partial \left[\vec{\Phi}(\Delta q ) \right]/ {\partial (\Delta q)}
        &= \frac{2\pi}{c}\left[\Delta f_1,\Delta f_1+\Delta f_2,\cdots\sum_{k=1}^{N-1} {\Delta f_k } \right]^T\notag\\
        &= \frac{2\pi}{c}\vec{\Gamma }\Delta \vec{f}
    \end{align}
   It follows that
    \begin{align}
        \Delta \vec{f}^T \vec{\Gamma}^T \vec{R}^{-1} \vec{\Phi}\left(\Delta q \right) &= 0 \notag\\
        \Delta \vec{f}^T \vec{\Gamma }^T \vec{R}^{-1} \left( \frac{2\pi \Delta q}{c}
                 \vec{\Gamma} \Delta \vec{f} - \vec{\Theta}_e \right)&= 0
    \end{align}
    Then
    \vspace{-14pt}
    \begin{align}
         \Delta q &= \frac{c}{2\pi} \frac{\Delta \vec{f}^T \vec{\Gamma}^T \vec{R}^{-1} \vec{\Theta}_e }
                  { \Delta \vec{f}^T \vec{\Gamma}^T \vec{R}^{-1} \vec{\Gamma}\Delta \vec{f}  } \\
         \vec{E} \left[\Delta q^2 \right]
                                      &= \frac{c^2}{4\pi ^2 }  \frac{  \Delta \vec{f}^T \vec{\Gamma }^T \vec{R}^{-1} \vec{E}\left[ \vec{\Theta}_e \vec{\Theta }_e^T \right]\vec{R}^{-1} \vec{\Gamma} \Delta \vec{f}   }
                                       {  (\Delta \vec{f}^T \vec{\Gamma}^T \vec{R}^{-1} \vec{\Gamma}\Delta \vec{f})^2  }
    \end{align}
    Since
    \begin{gather}
        \vec{E}\left[\vec{\Theta}_e \vec{\Theta}_e^T \right]=\sigma_\theta^2 (\vec{I}_{N-1}+\vec{uu}^T)=\sigma_\theta^2 \vec{R}
    \end{gather}
    Therefore, the MMSE for this moderate SNR is given by
    \begin{align}
    \label{equa:MMSE}
         \text{MMSE}&=\vec{E}\left[\Delta q^2 \right]
                      =\frac{c^2}{4\pi ^2} \frac{\sigma_\theta ^2} {\Delta \vec{f}^T \vec{\Gamma }^T \vec{R}^{-1} \vec{\Gamma}\Delta
                       \vec{f} } \notag \\
                    &= \frac{c^2}{4\pi ^2}  \frac{\sigma_\theta ^2}{\Delta \vec{f}^T \vec{\Gamma }^T
                         \left(\vec{I}_{N-1}-\vec{uu}^T/N \right) \vec{\Gamma}\Delta \vec{f} }
    \end{align}

    \subsection{CRB}
    \label{subsect:CRB}
     It is well known that the phase noise in (\ref{equa:LS_estimation}) follows the wrapped normal distribution due to modulo $2\pi$ operation\cite{Bahlmann:06}\cite{Cheng:11}. So the estimation is usually not unbias and the CRB does not exist. But the noise can be approximated as normal distribution under the assumption of high SNR.

     Consider the signal $y(k)= e^{j2\pi q / \lambda _k }+ n(k)$, where $q$ is the parameter to be estimated and  $n(k)$ is the complex Gaussian noise with zero-mean and variance $\vec{E}[ n^2(k)]=\sigma_n^2= \sigma ^2$ and the signal-to-noise ratio (SNR) is defined as $\textrm{SNR} =1/\sigma^2$. The signal can be expressed as
     \vspace{-0pt}
     \begin{gather}
        \hspace{-4pt} y(k) = e^{j2\pi q / \lambda _k }  + n(k)
             = e^{j2\pi q / \lambda _k } (1 + n(k)e^{-j2\pi q / \lambda _k } )
     \end{gather}
      \vspace{-0pt}
     Let $n'(k) = n(k)e^{-j2\pi q / \lambda_k } = n'_R (k) + n'_I (k)j$, then $n'(k)$ is statistically
     equivalent to $n(k)$ with $\vec{E}[n'^2 (k) ]=\sigma ^2$ and $\vec{E}[{n'_I}^2(k)] = {\sigma ^2/2}$. Therefore
     \begin{align}
        y(k)&=e^{j2\pi q /{\lambda _k }} (1+n'_R(k)+n'_I(k)j) \notag\\
            &= \sqrt { (1+n'_R(k))^2+{n'_I}^2(k) } e^{(j2\pi q /{\lambda _k }+\theta_e (k) )}
     \end{align}

      Where $\theta_e (k)$ is phase noise corresponding to $n(k)$. At high SNR, the following approximation holds:
     \begin{gather}
     \theta_e(k) \approx tan(\theta_e (k))= \frac{ n'_I (k)} {1 + n'_R (k)} \approx n'_I (k)\\
     \vec{E} \left[\theta_e^2 (k) \right] = \sigma_\theta^2  \approx {\sigma ^2 /2}
     \end{gather}

     Define $\vec{y} = [y(1),y(2),\cdots y(N) ]$ and $\vec{\varphi} =\nobreak [\varphi(1),\varphi(2),\cdots \varphi(N) ]$, $\varphi (k) ={2\pi q/\lambda_k+\theta _e(k)}$. It is clear that the original problem of estimating $q$ using $\vec{\varphi}$, corrupted by noise of variance ${\sigma^2 /2}$, is equivalent to estimating it from $\vec{y}$ with double variance. The Cramer-Rao bound of the equivalent problem is easily obtained as follows\cite{Rife:74}.

     Let  $a_k =\vec{Re} \{y(k)\}$, $b_k =\vec{Im} \{y(k)\}$. The probability distribution function is
     \begin{align}
        f(\vec{y},q)&=\hspace{-4pt}\left(\frac{1}{\sqrt {\pi \sigma_n^2} }\right)^{2N}
                         \hspace{-4pt}\exp \Bigg[\hspace{-4pt}- \frac{1}{\sigma _n^2}\sum_{k=1}^N  \bigg( \big(a_k- \cos ( 2\pi q /\lambda _k) \big)^2  + \big(b_k- \sin ( 2\pi q /\lambda _k) \big)^2  \bigg) \Bigg]  \notag\\
                       &=\hspace{-4pt}\left(\frac{1}{\sqrt {\pi \sigma_n^2} }\right)^{2N}
                         \hspace{-4pt}\exp \Bigg[\hspace{-4pt}- \frac{1}{\sigma _n^2}\sum_{k=1}^N  \bigg( \big(a_k- \cos ( 2\pi q f_k/c) \big)^2  + \big(b_k- \sin ( 2\pi q f_k/c) \big)^2  \bigg) \Bigg]
     \end{align}

     The entry of Fisher information matrix and the Cramer-Rao bound (CRB) are given by
     \begin{gather}
        [\vec{F}]_{q,q}=\vec{E}\left[ \frac{\partial \log f(\vec{y},q)}{\partial q} \frac{\partial \log f(\vec{y},q)}{\partial q} \right]=\frac{8\pi^2}{c^2 \sigma_n^2 } \sum_{k=1}^N {f_k^2 }
     \end{gather}
     \vspace{-12pt}
     \begin{gather}
     \label{equa:CRB}
        \mathrm{CRB}(q)= ([\vec{F}]_{q,q})^{-1}
                       =\frac{c^2 \sigma _n^2}{8\pi^2}\left(\sum_{k=1}^N
                       {f_k^2 } \right)^{-1}
                       = \frac{c^2 \sigma_\theta ^2}{4\pi^2}\left(\sum_{k=1}^N {f_k^2} \right)^{-1}
     \end{gather}
     Note that the CRB coincides with the HMSE, see (\ref{equa:HMSE}).

\section{Frequency Design}

    There exist three different frequency design goals, including maximizing UMR, minimizing the outlier probability and enhancing the MSE performance. We provide corresponding design criterion independently and then present a simple algorithm with all factors considered together.

    \subsection{Frequency Spacing Design for UMR and Outlier}
    \newtheorem {propo}{\noindent Proposition}
     \begin{propo}
       \label{propo:prime}
       Suppose the frequency resolution of the system is $\Delta f$, that is to say, each frequency spacing must be integer number of $\Delta f$. If all the frequency spacing can be expressed as $\Delta f_i=p_i \Delta f$, $p_i$ is prime number, then the UMR is $\Delta L=c/\Delta f$.
       \end{propo}
     \begin{proof}
         It is the corollary of Theorem \ref{theo:UMR}.
     \end{proof}
     Remark: Proposition \ref{propo:prime} demonstrates that we can extend the measurement range by simply reducing the GCD and rearranging the elements of $\Delta \vec{f}$ has no influence on the UMR.

     According to (\ref{equa:pdf_Multi}), it is clear that a sharp peak, comparable with the main peak at the true $q_0$, will appear in the probability density functions if the peaks belonging to different PDF $p(f_i,f_{i+1}|q)$ coincide at the same location, i.e. $\Delta q^{(i)}(k_i)= \Delta q^{(j)}(k_j)=\cdots$. Then the outlier probability will increase.
     To avoid the occurrence of outlier, we provide the following proposition.

    \begin{propo}
       \label{propo:outlier}
           Assume the GCD of the frequency spacing $\Delta \mathbf{f}$ is $\Delta f$, i.e., $\Delta f_i  = n_i \Delta f$, $n_i$ is positive integer. If $n_i$ and $n_j$ are co-prime with $\mathrm{GCD}(n_i,n_j)=1$ for any two frequency spacing $\Delta f_i,\Delta f_j$, then for any $\Delta q^{(i)}(k_i),\Delta q^{(j)}(k_j)<c/\Delta f$, $\Delta q^{(i)}(k_i)\neq \Delta q^{(j)}(k_j)$ ($\Delta q^{(i)}(k_i)$ is the location of possible sharp peak in $p(f_i,f_{i+1}|q)$). In other word, sharp peak of different PDF $p(f_i,f_{i+1}|q)$ will not coincide and strengthen at the same location.
    \end{propo}

    \begin{proof}
      Suppose $\Delta q^{(i)}(k_i)= \Delta q^{(j)}(k_j)$, then
     \begin{gather*}
         k_ic/\Delta f_{i}=k_jc/\Delta f_{j}\\
        \Rightarrow{k_i}{n_j}={k_j}{n_i}
     \end{gather*}

     Since $\mathrm{GCD}(n_i,n_j)=1$, then $k_i$ must have the factor $n_i$. In other word, $k_i=m_in_i$, $m_i\geq 1$ is positive integer. Similarly, $k_j=m_jn_j$.
     On the other hand, $\Delta q<c/\Delta f$, we get
     \begin{gather*}
        k_i<n_i,\quad k_j<n_j
     \end{gather*}
     Contradiction.
    \end{proof}

    \subsection{Frequency Spacing Design for MSE}
    For $f_i\gg B$, which is typical in most measurement system, the HMSE expressed in (\ref{equa:HMSE}) is not sensitive to different frequency pattern. In other words, different frequency design methods have negligible influence on the HMSE performance. Therefore, we will be concerned with the optimization of MMSE rather than HMSE, in the form of
     \begin{gather}
     \label{equa:MMSE}  
     \underset{\scriptstyle \Delta \vec{f}}{\max}(\Delta \vec{f})^T \vec{\Gamma }^T
                    \left(\vec{I}_{N-1}-\vec{uu}^T/N \right) \vec{\Gamma}\Delta \vec{f}
     \end{gather}

    \begin{propo}
    \label{propo:increase_B}
        Suppose $B_2>B_1$, for any given measurement frequency in frequency band $B_1$ with spacing $\Delta \vec{f} =[\Delta f_{1} ,\Delta f_2 ,\cdots\Delta f_{N-1} ]^T$, If we pick a new set of frequencies from band $B_2$ in such a manner that $\Delta\vec{f'}=\nobreak[\Delta f'_{1} ,\Delta f'_2 ,\cdots\Delta f'_{N-1} ]^T$, $\Delta f'_k \geq \Delta f_k$, then the new measurement frequencies will result in high ranging accuracy.
    \end{propo}

    \begin{proof}
        Let
        \begin{align}
            \Delta \vec{f'} &= [\Delta f'_{1},\Delta f'_2 ,\cdots\Delta f'_{N-1} ]^T  \notag\\
            &= [\Delta f_{1},\Delta f_2 ,\cdots\Delta f_{N-1}  ]^T + [d_{1} ,d_2 ,\cdots d_{N-1} ]\notag\\
            &= \Delta \vec{f} + \vec{d}
        \end{align}
        Where $\vec{d}=[d_i|i=1,2,\cdots N-1, d_i\geq 0]$. Define $\vec{Q}=\nobreak\vec{\Gamma }^T \left( N\vec{I}_{N-1}- \vec{uu}^T \right)\vec{\Gamma }$, then
         \begin{align}
           \label{equa:Qij}
             \vec{Q}(i,j)&=
               \begin {cases}  {\left( {N - i} \right)j} & {i > j}  \\
                               {\left( {N - i} \right)i} & {i = j}  \\
                               {\left( {N - j} \right)i} & {i < j}
               \end {cases}
               ,\vec{Q}(i,j)>0
         \end{align}
        It follows that
        \begin{align}
            (\Delta \vec{f'})^T &\vec{\Gamma }^T \left(N\vec{I}_{N-1}-\vec{uu}^T \right) \vec{\Gamma}\Delta \vec{f'} \notag\\
            =&(\Delta \vec{f})^T \vec{\Gamma }^T \left(N\vec{I}_{N-1}-\vec{uu}^T \right) \vec{\Gamma}\Delta \vec{f}
               +\vec{d}^T \vec{\Gamma }^T \left(N\vec{I}_{N-1}-\vec{uu}^T \right) \vec{\Gamma} \vec{d} \notag\\
            >&(\Delta \vec{f})^T \vec{\Gamma }^T \left(N\vec{I}_{N-1}-\vec{uu}^T \right) \vec{\Gamma}\Delta \vec{f}
        \end{align}

    \end{proof}

   Remark: For equal-spaced measurement frequency such as RIPS, let the frequency spacing be $\Delta{\vec{f}}=[\Delta f,\Delta f,\cdots\Delta f]^T$, it is easily verified that
   \begin{align*}
    \vec{E}[(q-q_0)^2]&=\frac{c^2}{4\pi^2}\frac{\sigma _\theta^2}
    {\Delta \vec{f}^T \vec{\Gamma }^T
                    \left(\vec{I}_{N-1}-\vec{uu}^T/N \right) \vec{\Gamma}\Delta \vec{f} }\\
                        &=\frac{c^2 12\sigma _\theta ^2}{4\pi ^2 \Delta f^2 N( N^2-1)}\\
                        &=\frac{c^2 12\sigma _\theta ^2 ( N-1)} {4\pi ^2 B^2 N(N+1)}
   \end{align*}
   This means that the ranging accuracy of measurement method using equal-spaced frequency increases in proportion to $N$ as well as $B^2$.

   \newtheorem {lemma}{\noindent Lemma}
   \begin{lemma}
    \label{lemma:partial_sum}

     Let $\vec{a}=\left\{a_k,1\le k\le N\mid 0<a_1<\dots<a_N\right\}$ be a set of positive numbers sorted in ascending
     order and  $\vec{g}=\pi(\vec{a})=\left[g_k,1\le k\le N\right]$\,be a permutation of $\vec{a}$. Denote the set of
     all the permutations of $\vec{a}$ as $\mathcal{A}$ and define the partial sums sequence $\vec{b}$ of $\vec{g}$ as
      \begin{gather*}
        b_k= \sum_{i=1}^k g_i ,1\le k\le N, \hspace{2pt}   \overline{\vec{b}}=\frac{1}{N} \sum_{k=1}^N b_k
      \end{gather*}
      \vspace{-12pt}
      \begin{align}
      V(\vec{b})&=\frac{1}{N}\sum_{k=1}^N (b_k-\overline{\vec{b}})^2\notag\\
                   &=\frac{1}{N}\sum_{k=1}^N \left(\sum_{i=1}^k g_i-\frac{1}{N}\sum_{k=1}^N\sum_{i=1}^k g_i\right)^2\notag\\
                   &=f(\vec{g})
      \end{align}
     where $ \overline{\vec{b}} $ and $ V(\vec{b}) $ are the mean and variance of $ \vec{b} $. The optimal permutation
     is defined as ${\vec{g}}^{\ast} $, satisfying
       \begin{gather}
          \vec{g}^{\ast}=\arg {\max_{\vec{g} \in \mathcal{A}}} f(\vec{g})=\arg {\max_{\vec{g} \in \mathcal{A}}}
                          \frac{1}{N}\sum_{k=1}^N \left( b_k-\overline{\vec{b}} \right)^2
       \end{gather}

     Then the optimal sequences are
         \begin{equation}\vec{g}^{\ast}\hspace{-2pt}=\hspace{-2pt}
            \begin {cases} [a_1,a_3,a_5,\dots a_N,a_{N-1},\dots a_4,a_2 ] & N\hspace{2pt} is\hspace{2pt} odd\\
                           [a_1,a_3,a_5,\dots a_{N-1},a_N,\dots a_4,a_2 ] &\hspace{-2pt}N\hspace{2pt} is\hspace{2pt} even
            \end {cases}
         \end{equation}
     and
         \begin{equation}\vec{g}^{\ast}\hspace{-2pt}=\hspace{-2pt}
            \begin {cases} [a_1,a_2,a_4,\dots a_{N-1},a_N,\dots a_5,a_3 ] & N\hspace{2pt} is\hspace{2pt} odd\\
                          [a_1,a_2,a_4,\dots a_N,a_{N-1},\dots a_5,a_3 ] & \hspace{-2pt}N\hspace{2pt} is\hspace{2pt} even
            \end {cases}
         \end{equation}
    \end{lemma}
    \begin{proof}
       The author proved it in \cite{Wei:12}.
    \end{proof}

    \begin{lemma}
    \label{lemma:MMSE}
     Suppose the symbols $\vec{a},\vec{g},\mathcal{A},\vec{g}^{\ast},f(\vec{g})$ are defined as lemma.
     \ref{lemma:partial_sum}. Let $\vec{g}^{(1)}=[g_{k+1},\;1 \leq k \leq N-1]$ and $\vec{a}_s=\{a_{k+1},\;1 \leq k \leq
     N-1\}=\{a_2,a_3 \cdots a_N \}$ denote the one-bit left-shift sequence of $\vec{g}$ and the subset of $\vec{a}$. Define
     all the permutations of  $\vec{a}_s$ as $\mathcal{V}=\{\vec{v}\left| {\vec{v}=\pi(\vec{a}_s)} \right\}$, then we
     have
        \begin{enumerate}
        \renewcommand{\baselinestretch}{1.2}\normalsize
         \item	For \,any $\vec{g}$ and its left-shift sequence $\vec{g}^{(1)}$,
                \begin{gather}
                    f(\vec{g})=\left(\vec{g}^{(1)}\right)^T \vec{\Gamma }^T \left(\vec{I}_{N-1}-\vec{uu}^T/N \right) \vec{\Gamma g}^{(1)}
                \end{gather}

         \item	The optimal permutation $\vec{v}^{\ast}$ that maximizing $(\vec{v})^T \vec{\Gamma }^T
                \left(\vec{I}_{N-1}-\vec{uu}^T/N \right) \vec{\Gamma v}$ is just the one-bit left-shift sequence of
                $\vec{g}^{\ast}$,
         \begin{align*}
          \vec{v}^{\ast}&=\arg \max_{\vec{v} \in \mathcal{V}}(\vec{v})^T \vec{\Gamma }^T
                    \left(\vec{I}_{N-1}-\vec{uu}^T/N \right) \vec{\Gamma v}\\
                           &=\{ g_{k + 1}^{\ast} ,1 \leq k \leq N-1 \} = \left(\vec{g}^{\ast} \right)^{(1)}
         \end{align*}
         \vspace{-8pt}
         \begin{equation}\vec{v}^{\ast}\hspace{-2pt}=\hspace{-2pt}
           \begin {cases} [a_3,a_5,\dots a_N,a_{N-1},\dots a_4,a_2 ] & N\hspace{2pt} is\hspace{2pt} odd\\
                          [a_3,a_5,\dots a_{N-1},a_N,\dots a_4,a_2 ] &\hspace{-2pt}N\hspace{2pt} is\hspace{2pt} even
           \end {cases}
         \end{equation}
         and
         \begin{equation}\vec{v}^{\ast}\hspace{-2pt}=\hspace{-2pt}
           \begin {cases} [a_2,a_4,\dots a_{N-1},a_N,\dots a_5,a_3 ] & N\hspace{2pt} is\hspace{2pt} odd\\
                          [a_2,a_4,\dots a_N,a_{N-1},\dots a_5,a_3 ] & \hspace{-2pt}N\hspace{2pt} is\hspace{2pt} even
           \end {cases}
         \end{equation}
       \end{enumerate}
    \end{lemma}
     \begin{proof}  
     \begin{enumerate}
       \item For any $\vec{g}$ and its partial sums sequence $\vec{b}$, let $\vec{d}=[d_1,d_2,\cdots,d_N]^T$ with $d_k=b_k-g_1$, and $\overline{\vec{d}}=\nobreak\frac{1}{N} \sum_{k=1}^N d_k=\frac{1}{N} \sum_{k=2}^N d_k$, then
           \begin{align*}
              f(\vec{g})&=\sum_{k=1}^N \left(b_k-\frac{1}{N} \sum\nolimits_{j=1}^N b_j\right)^2 \\
                           &=\sum_{k=1}^N \Big( b_k-( Ng_1+(N-1)g_2+\cdots g_N )/N \Big)^2 \\
                           &=\left[ \sum_{k=2}^N  \left(d_k-\overline{\vec{d}} \right)^2 +\overline{\vec{d}}^2\right]
           \end{align*}
           Define $\vec{d}^{(1)}=[d_2 ,d_3 ,\cdots,d_N]^T$ and $\vec{g}^{(1)}=[g_2 ,g_3 ,\cdots,g_N]^T$, we can get
           \begin{gather*}
            \vec{d}^{(1)}=\vec{\Gamma }\vec{g}^{(1)},\quad\overline{\vec{d}}=\frac{\vec{u}^T\vec{\Gamma
                             }\vec{g}^{(1)}}{N}
           \end{gather*}

           \vspace{-10pt}

           \begin{align*}%
            \left(\vec{d}^{(1)}-\overline{\vec{d}}\vec{u} \right)&=\vec{\Gamma }\vec{g}^{(1)}-\frac{\vec{u}^T\vec{\Gamma}\vec{g}^{(1)}\vec{u}}{N}\\
            &=\left(\vec{I}_{N-1}-\frac{\vec{uu}^T}{N} \right) \vec{\Gamma g}^{(1)}\\
            \left[ \sum_{k=2}^N  \left(d_k-\overline{\vec{d}} \right)^2 \right]&= \left(\vec{d}^{(1)}-\overline{\vec{d}}\vec{u} \right)^T \left(\vec{d}^{(1)}-\overline{\vec{d}}\vec{u} \right)\\
            &=\left(\vec{g}^{(1)}\right)^T \vec{\Gamma }^T\left(\vec{I}_{N-1}-\vec{uu}^T/N \right)
               \left(\vec{I}_{N-1}-\vec{uu}^T/N \right)\vec{\Gamma g}^{(1)}
            \end{align*}

            \vspace{-14pt}

            \begin{align*} 
            f(\vec{g})&=\left(\vec{g}^{(1)}\right)^T \vec{\Gamma }^T \left(\left(\vec{I}_{N-1}-\vec{uu}^T/N \right)
                         \left(\vec{I}_{N-1}-\vec{uu}^T/N \right) +\vec{uu}^T/N^2 \right) \vec{\Gamma g}^{(1)}\\
                      &=\left(\vec{g}^{(1)}\right)^T \vec{\Gamma }^T\left(\vec{I}_{N-1}-\vec{uu}^T/N \right)\vec{\Gamma g}^{(1)}
           \end{align*}

       \item According to lemma \ref{lemma:partial_sum}, $g_1^{\ast}=a_1$, then $\left(\vec{g}^{\ast} \right)^{(1)}\hspace{-2pt} \in \hspace{-2pt}\mathcal{V}$, we have
                \begin{align}
                \label{equa:large}
            f(\vec{g}^{\ast})&=\left({\vec{g}^{\ast}}^{(1)}\right)^T \vec{\Gamma }^T\left(\vec{I}_{N-1}-\vec{uu}^T/N \right)\vec{\Gamma }{\vec{g}^{\ast}}^{(1)}\notag\\
               &\leq \left(\vec{v}^{\ast}\right)^T \vec{\Gamma }^T\left(\vec{I}_{N-1}-\vec{uu}^T/N \right)\vec{\Gamma }\vec{v}^{\ast}
               \end{align}
           Let $\vec{\tilde v}^{\ast}=\left\{a_{1},\vec{v}^{\ast}\right\}
           = \left\{a_{1},v_1^{\ast},\cdots v_{N-1}^{\ast}  \right\}$, $\vec{\tilde v}^{\ast} \in \mathcal{V}$, therefore
              \begin{gather}
              \label{equa:small}
                \hspace{-12pt}\left(\vec{v}^{\ast}\right)^T \vec{\Gamma }^T\left(\vec{I}_{N-1}-\vec{uu}^T/N \right)\vec{\Gamma }\vec{v}^{\ast}=f(\vec{\tilde v}^{\ast}) \leq f(\vec{g}^{\ast})
              \end{gather}
              From (\ref{equa:large}) and (\ref{equa:small}), it is readily seen that
              \begin{gather} %
                \left(\vec{v}^{\ast}\right)^T \vec{\Gamma }^T\left(\vec{I}_{N-1}-\vec{uu}^T/N \right)\vec{\Gamma }\vec{v}^{\ast}= f(\vec{g}^{\ast})\notag\\
                \vec{v}^{\ast}= \left( \vec{g}^{\ast} \right)^{(1)}
              \end{gather}
     \end{enumerate}
     \end{proof}

    Based on lemma \ref{lemma:MMSE}, we deduce the following Theorem

    \begin{theorem}
    \label{theo:optimal_permutation}
     Suppose that $\Delta \vec{f}=[\Delta f_1,\Delta f_2,\cdots\Delta f_{N-1}]$ is any given frequency spacing sorted in
     ascending order, then the optimal permutation of $\Delta \vec{f}$, in the sense of maximizing equation(\ref{equa:MMSE}),
     takes the form
     \begin{gather}
        \Delta \vec{f}^{\ast}=[\Delta f_1,\Delta f_3,\cdots \Delta f_{N-1},\Delta f_{N-2},\cdots \Delta f_4,\Delta f_2]
     \end{gather}
     \end{theorem}

     \begin{proof}
         It is a direct conclusion of lemma \ref{lemma:MMSE}.
     \end{proof}

    We named the permutation of the form $\Delta \vec{f}^{\ast}$ provided in Theorem \ref{theo:optimal_permutation} as min-error permutation and the dual form $\Delta \vec{\tilde f}=[ \Delta f_{N-1} , \cdots \Delta f_3 ,\Delta f_1 ,\Delta f_2,\Delta f_4 , \cdots \Delta f_{N-2}]$ as max-error permutation. It is proved in \cite{Eilon:77} that the permutation in the form of $\Delta \vec{\tilde f}$ has extremely small variance once substituted into equation(\ref{equa:MMSE}). This implies that this permutation will lead to large ranging error when used in measurement frequency design.

    \begin{theorem}
    \label{theo:optimal_frequency}
      For a measurement system with bandwidth $B$, frequency number $N$, frequency resolution $\Delta f_{\min}$, the total  frequency number $\lfloor M={B/\Delta f_{\min}} \rfloor$ ($\lfloor \, \rfloor$ \,denotes the floor function), if the range  fall in $[-c/2B, c/2B]$, then the optimal measurement frequencies must be picked up from as near as possible to both ends of the frequency band, with frequency spacing in the form of
        \begin{gather}
          \Delta \vec{f}^{\ast}=[1,1,\cdots M+2-N,\cdots 1,1]^T \Delta f_{\min}
        \end{gather}
    \end{theorem}

     \begin{proof}
       See Appendix \ref{append:optimal_frequency}.
     \end{proof}

    \subsection{Algorithm}
    \label{subsect:algorithm}
    Based on Proposition \ref{propo:prime}$\sim$\ref{propo:increase_B} and Theorem \ref{theo:optimal_permutation}, we proposed a quite simple yet effective algorithm in Table \ref{table_algorithm}:

    \begin{table}[!t]
    \renewcommand{\arraystretch}{1.3}  
    \caption{The prime-based min-error method}
    \label{table_algorithm}
    \centering      
    \begin{tabular}{|p{9cm}|}  
    \hline
    $\star$ Define the set of primes less than $M$ as $\mathcal{P}_{M}$ (i.e. $\mathcal{P}_{10}=\{2,3,5,7 \}$). The positive
           integer $M$ is selected so that $|\mathcal{P}_{M}|\gg N$. \\

    $\star$  Assume $\Delta f_{\min}$ is the frequency resolution of measurement
           system and $K$ is the common factor, i.e., $K\Delta f_{\min}$ is the minimal
           frequency spacing of the frequency. To make full use of the
           bandwidth $B$, we perform the following steps (assume the UMR requirement is known as priori information or the UMR is set as needed) : \\

      \begin{enumerate}
        \item $i=1$.
        \item Pick the $N-1$ consecutive primes subset $\mathcal{S}^i=\{\mathcal{P}_M (i)\cdots \mathcal{P}_M (N+i-2)\}$
            and find the corresponding $K$, obey
               $K\Delta f_{\min} \sum_{j=1}^{N-1} \mathcal{S}^i(j) \leq B $
           and
               $(K+1)\Delta f_{\min}\sum_{j=1}^{N-1} \mathcal{S}^i(j) > B $.
        \item if $c/(K\Delta f_{\min})>UMR$, go to step 4); else, $i=i+1$ and go to step 2).
        \item  Then, the frequency spacing set $\Delta{\mathcal{F}}=\{ \Delta f_j=K\Delta f_{\min} \mathcal{S}^i(j) \left| {j=1 \cdots N-1}\right \}$ satisfying $\sum_{j=1}^{N-1}{\Delta f_j}\leq B$ is constructed.
      \end{enumerate}

   $\star$  Sort the frequency spacing in ascending order and obtain
       $\Delta \vec{f}=[\Delta f_1,\Delta f_2,\cdots\Delta f_{N-1}]$. \\
     \begin{itemize}
       \item If $N$ is odd, $\Delta \vec{f}$ can be rearranged as
           $$\Delta \vec{f}^{\ast}=[\Delta f_1,\Delta f_3,\cdots\Delta f_{N-2},\Delta f_{N-1},\cdots \Delta f_4,\Delta f_2]$$
       \item If $N$ is even, then
           $$\Delta \vec{f}^{\ast}=[\Delta f_1,\Delta f_3,\cdots\Delta f_{N-1},\Delta f_{N-2},\cdots \Delta f_4,\Delta f_2]$$
     \end{itemize}

   $\star$  The measurement frequencies is finally obtained
     $f_i=f_1 +\sum_{k=1}^{i-1} \Delta \vec{f}^{\ast} (k), i=1,2,\cdots N$.  \vspace{3mm}\\

   \hline
   \end{tabular}
   \end{table}
 Remark: The large enough primes set in Table \ref{table_algorithm} allow the algorithm to select the appropriate primes as needed. In the latter part of the paper, the notation $(B,N,\Delta f_{\min},i,K)$ stands for the design parameters of the min-error method defined above. Note that the UMR of the method is $\Delta L\approx c/(K\Delta f_{\min})$. The common factor $K$ and the index $i$ have great influence on the UMR as well as the estimation accuracy, as is shown in the following simulation. Here, $K$ is set as large as possible while $i$ is set in the opposite direction, see Fig.\ref{fig_impact_K} for the reason.

\section{Simulation Results}

   In this section, we present simulation results to compare different frequency design methods under various scenarios. The measurement frequencies of the Towers method satisfy\cite{Towers:03}
    \begin{gather*}
      \frac{f_N-f_{N-2}}{f_N-f_{N-1}}=\frac{f_N-f_{N-3}}{f_N-f_{N-2}}=\cdots=\frac{f_N}{f_N-f_1}=\frac{f_N}{B} \\
      \frac{f_N}{f_N-f_i}=\left(\frac{f_N}{B}\right)^{i},f_i=f_N-f_N\left(\frac{B}{f_N}\right)^{i},i=1,2\cdots N-1
   \end{gather*}
   The measurement frequencies of RIPS are \cite{Kusy:06}
    \begin{equation*}
      f_i=f_1+\left(i-1\right)B/\left(N-1\right),i=1,2\cdots N
    \end{equation*}
   The frequencies of the constrained optimal method are given by
    \begin{gather*}
      \Delta \vec{f}=[1,1,\ldots,M+2-N,\ldots,1,1]^T\Delta f_{\min} \\
      f_i=f_1+\sum_{k=1}^{i-1}\Delta \vec{f}(k),M=B/\Delta f_{\min}
   \end{gather*}
   The frequencies of the prime-based min-error method are selected according to Table \ref{table_algorithm} (assuming N to be even)
    \begin{gather*}
    \Delta \vec{f}=[\Delta f_1,\Delta f_3,\cdots\Delta f_{N-1},\Delta f_{N-2},\cdots\Delta f_4,\Delta f_2]\\
    f_i=f_1+\sum_{k=1}^{i-1}\Delta \vec{f}(k)
    \end{gather*}
   To illustrate the effect of permutation, we also show the prime-based max-error method (having the same set of frequency spacing as the prime-based min-error method but with the worst permutation), in the form of
    \begin{gather*}
    \Delta \vec{f}=[\Delta f_{N-1},\Delta f_{N-3},\cdots\Delta f_1,\Delta f_2,\cdots\Delta f_{N-4},\Delta f_{N-2}]\\
    f_i=f_1+\sum_{k=1}^{i-1}\Delta \vec{f}(k)
    \end{gather*}

   Unless otherwise mentioned, the following parameter values are assumed: The measurement frequencies range from 400\,MHz to 500\,MHz, the initial frequency $f_1=400\textrm{\,MHz}$ is used. The interferometric range is $q_0=0\,\textrm{m}$, the phase error in each frequency is modeled as independent and identically distributed (i.i.d.) zero-mean complex white Gaussian noise with variance $E\{\theta_{e}^{2}(k)\}=\sigma_{\theta}^{2}$. The signal-to-noise ratio is defined as $\textrm{SNR}=1/(2\sigma_{\theta}^{2})$ according to section \ref{subsect:CRB}. We adopt the LS search algorithm to find the optimal solution, with a step size of 0.001\,m ranging from $-c/2\Delta f$ to $c/2\Delta f$, where $\Delta f=B/(N-1)$ and $c/\Delta f$ are the frequency separation and UMR of RIPS method, respectively.

   The MSE of different methods are averaged over 2000 Monte Carlo runs for each SNR. The MMSE of RIPS method, Prime-based min-error method and Prime-based max-error method are denoted by MMSE-RIPS, MMSE-min-error and MMSE-max-error. The CRB of all the three methods are nearly not distinguishable, only the CRB bound of RIPS method is plotted using (\ref{equa:CRB}) and denoted as CRB-RIPS.

   For a fair comparison among different frequency design methods, the same number of frequencies $N$ and bandwidth $B$ have been used. The notation $(B,N,\Delta f_{\min},i,K)$ is defined in section \ref{subsect:algorithm} with  $\Delta f_{\min}=65\,\textrm{Hz}$ is assumed, which is the frequency resolution of MICA2 platform used in the field experiment of section \ref{sect:experiment}\cite{Maroti:05}. With no information about UMR requirement, $i=1$ is used in the simulation except Fig.\ref{fig_impact_K}. Note that even in this case, the UMR is $\Delta L\approx c/(K\Delta f_{\min})>(c/B)\sum_{j=1}^{N-1} \mathcal{S}^1(j)\gg (c/B)(N-1)=\Delta L_{RIPS} $ with $\Delta L_{RIPS} $ denotes the UMR of RIPS, see Table \ref{table_algorithm}.


   \begin{figure}[!t]  %
      \centerline{
        \subfloat[Case I]{\includegraphics[height=2.8in,width=3in]{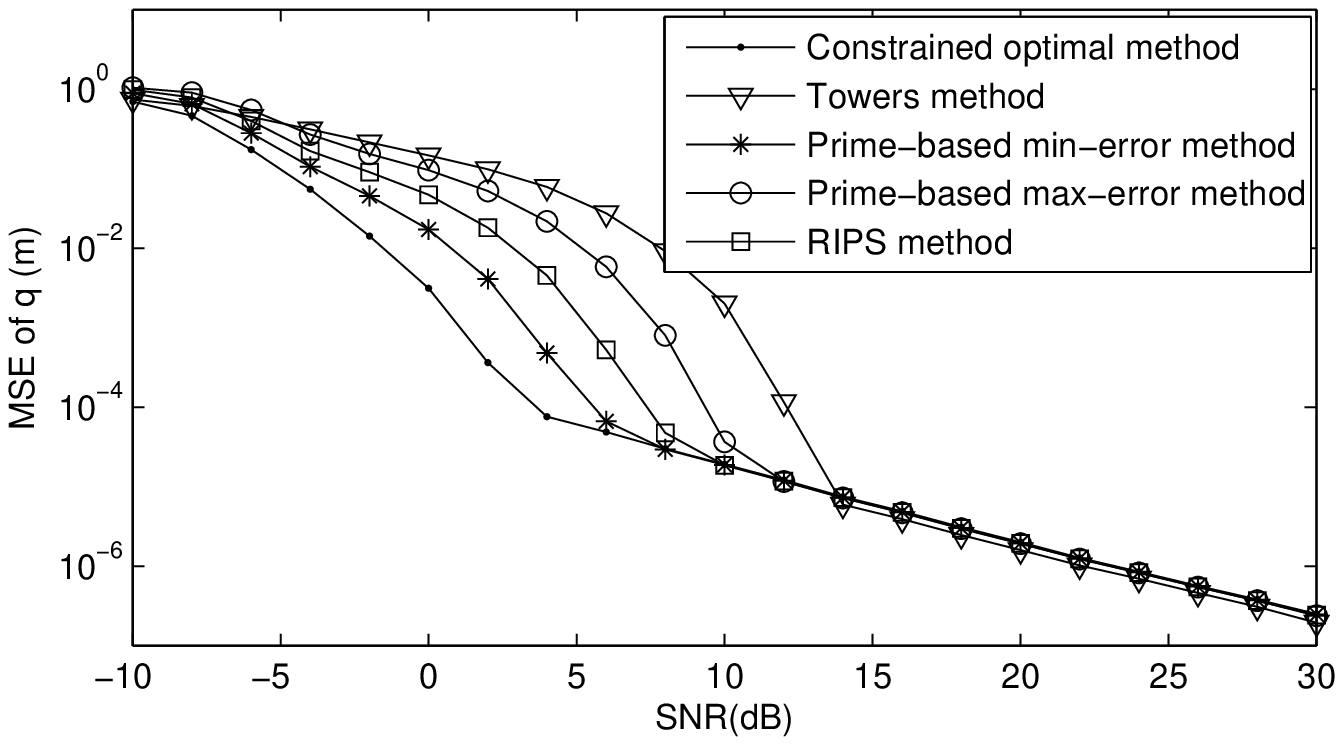}
        \label{fig_constrained}}
        \subfloat[Case II]{\includegraphics[height=2.8in,width=3in]{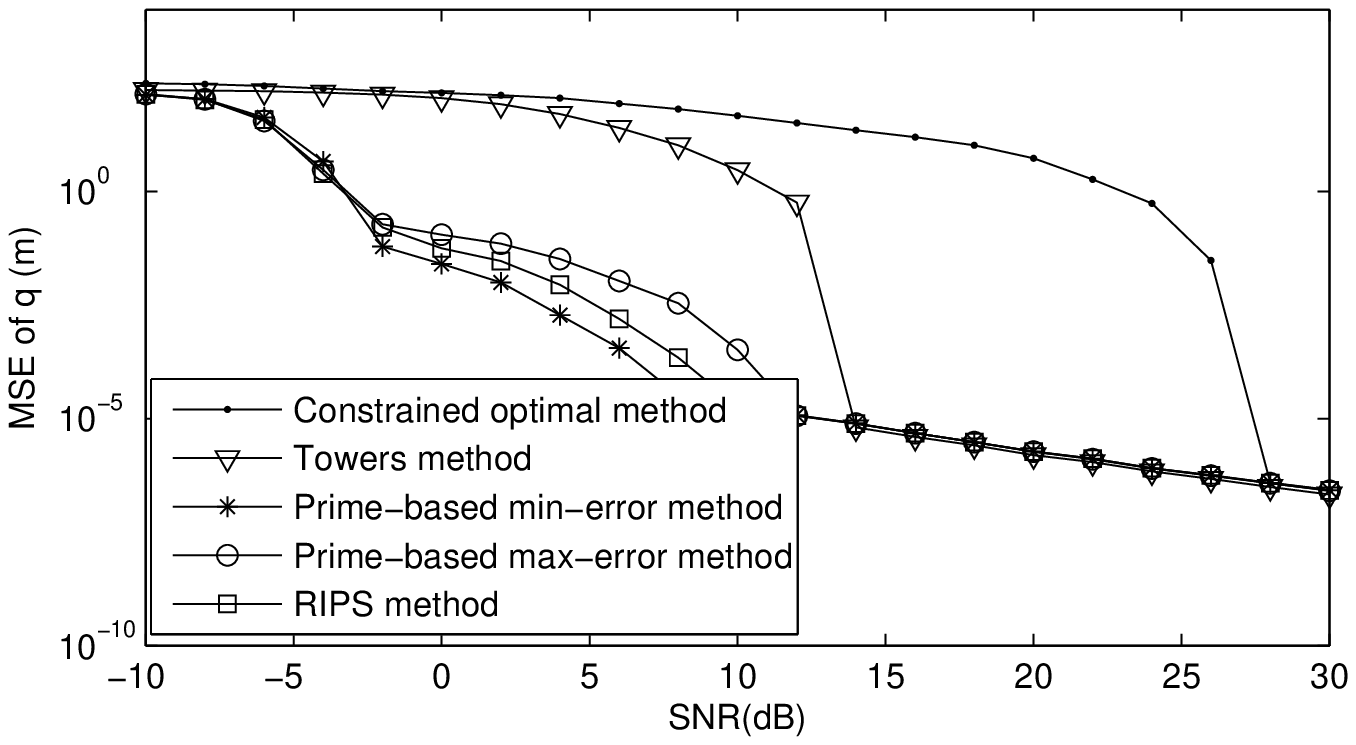}%
        \label{fig_unconstrained}}
        }
      \hfill
      \caption{The MSE versus measurement range with $N=31,B=90\,\textrm{MHz}$. (a) the search range of $[-c/2B, c/2B]$. (b) the search range of $[-c/2\Delta f, c/2\Delta f]$}%
      \label{fig_range}
   \end{figure}

    Fig.\ref{fig_constrained}, Fig.\ref{fig_unconstrained} illustrate that the constrained optimal method has the best ranging accuracy if the priori information that $q_0$ is in the range $[-c/2B, c/2B]$ is provided. It will fail to work once $q_0$ is outside this region. This feature is predicted in Theorem \ref{theo:optimal_frequency} and reduces its measurement range greatly. The proposed min-error method outperforms all the others and is slightly inferior to the constrained optimal method only in the  above limited measurement range. The MSE performance of towers method is not satisfactory due to the fact it uses a local instead of global method to optimize its frequency.

    \begin{figure}[!t]
      \centering
      \includegraphics[height=2.8in,width=3.5in]{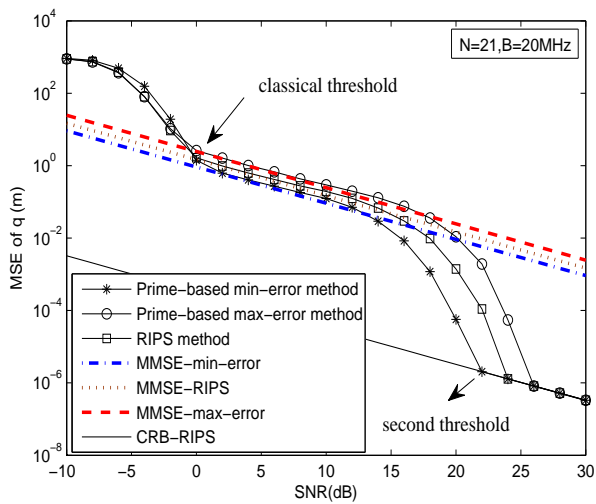}
      \caption{The double MSE curve for $N=21$ and $B=20\,\textrm{MHz}$. }
      \label{fig_double_mse}
    \end{figure}

    \begin{figure}[!t]
      \centering
      \includegraphics[height=2.8in,width=3.5in]{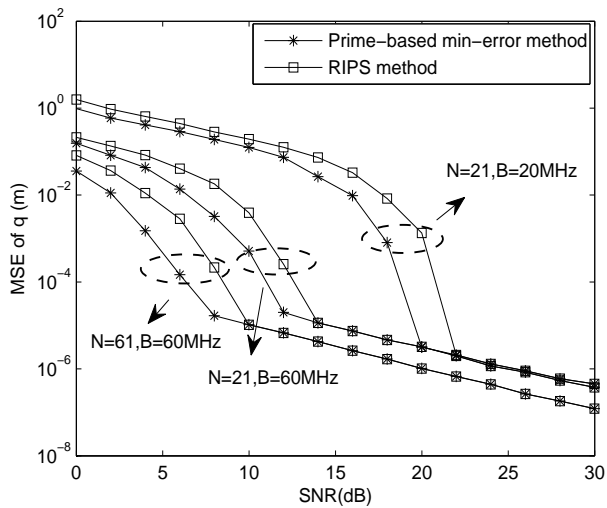}
      \caption{The impact of N and B.}
      \label{fig_impact_NB}
    \end{figure}

   When the condition $f_1\gg B$ is met, the multi-frequency interferometry ranging technique exhibits the unique double-threshold feature, as shown in Fig.\ref{fig_double_mse}. That is to say, besides the classical threshold, another threshold occurs. The MSE curve firstly follows the MMSE derived in (\ref{equa:MMSE}) tightly once SNR exceeds the classical threshold and drops to the CRB (or HMSE) when SNR reaching the second threshold. Only the CRB of RIPS is plotted for minor difference between the three methods and that is just the reason of optimizing MMSE in the paper. The prime-based min-error method has the minimal MMSE as well as the best estimation accuracy as expected. Compared with the prime-based max-error method, the merit of frequency spacing rearrangement is clearly visible, with about 4 dB gain obtained at the MSE of $10^{-4}$. This point will be further illustrated in later simulation.

   From Fig.\ref{fig_impact_NB}, it is observed that increasing the number of frequency or the measurement bandwidth will both improve the ranging accuracy of all the methods. The performance improvement achieved by increasing bandwidth $B$ is more significant, relative to the increment of $N$. Note that the proposed prime-based min-error method also works well for relatively wide bandwidth, also seen in Fig.\ref{fig_range}, although the design principle is derived under the assumption of $f_1\gg B$.

   \begin{figure}[!t]  
     \centering
     \includegraphics[height=2.8in,width=3.5in]{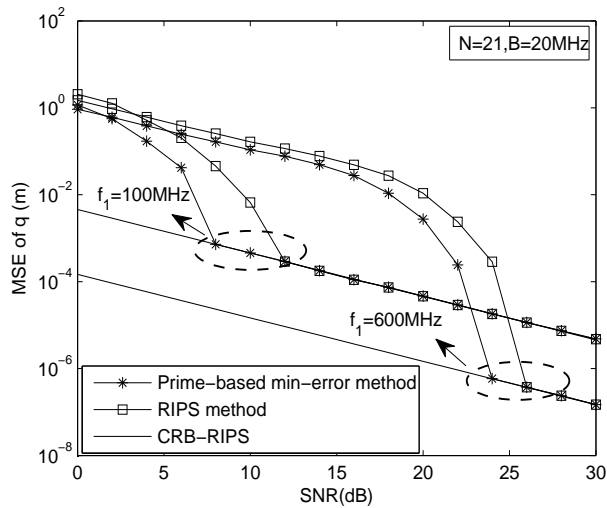}
     \caption{The impact of the initial measurement frequency.}
     \label{fig_initial_frequency}
   \end{figure}

   \begin{figure}[!t]
     \centering
     \includegraphics[height=2.8in,width=3.5in]{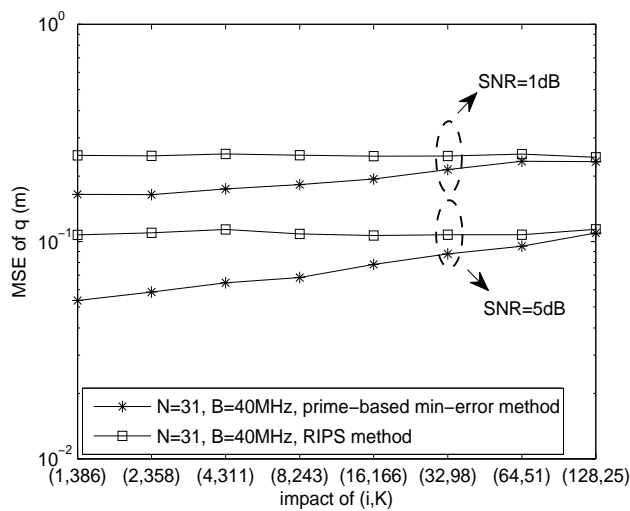}
     \caption{The impact of the common factor $K$ and the prime index $i$ in algorithm \ref{subsect:algorithm}.}
     \label{fig_impact_K}
   \end{figure}

   With fixed $B$ and $N$ in Fig.\ref{fig_initial_frequency}, it is interesting to find that the CRB is more easily attained at the cost of larger CRB for a low initial frequency $f_1$. The opposite has be seen for a high initial frequency. This phenomenon reveals an important design criterion that we prefer to use low initial frequency for moderate SNR and narrow bandwidth since the CRB is already attainable within this SNR region.

   The influence of the parameter $i,K$ on the MSE performance for a particular SNR is shown in Fig.\ref{fig_impact_K}, with $N=31$ and $B=40\,\textrm{MHz}$. The $x$-coordinate is the combination of $i,K$, which is the abbreviation of $(B,N,\Delta f_{\min},i,K)$ for simplicity($i$ is always set to 1 except in this case). The superiority of the min-error method is clearly seen for large $K$ and will diminish as $K$ decreases. The reason is that the difference between the frequency separation of the proposed method decreases and the frequency pattern tends to approach the RIPS method with decreasing $K$.

    \begin{figure}[!t]
       \centering
       \includegraphics[height=2.8in,width=3.5in]{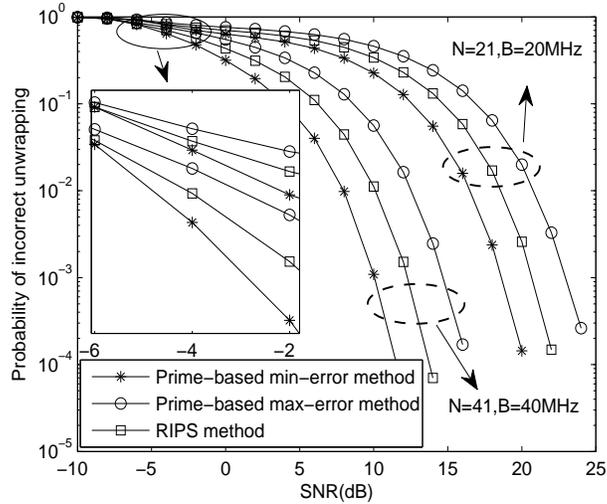}
       \caption{The probability of incorrect phase unwrapping.}
       \label{fig_false_unwrapping}
    \end{figure}

    \begin{figure}[!ht]  
      \centerline{
         \subfloat[Case I]{\includegraphics[height=2.8in,width=3in]{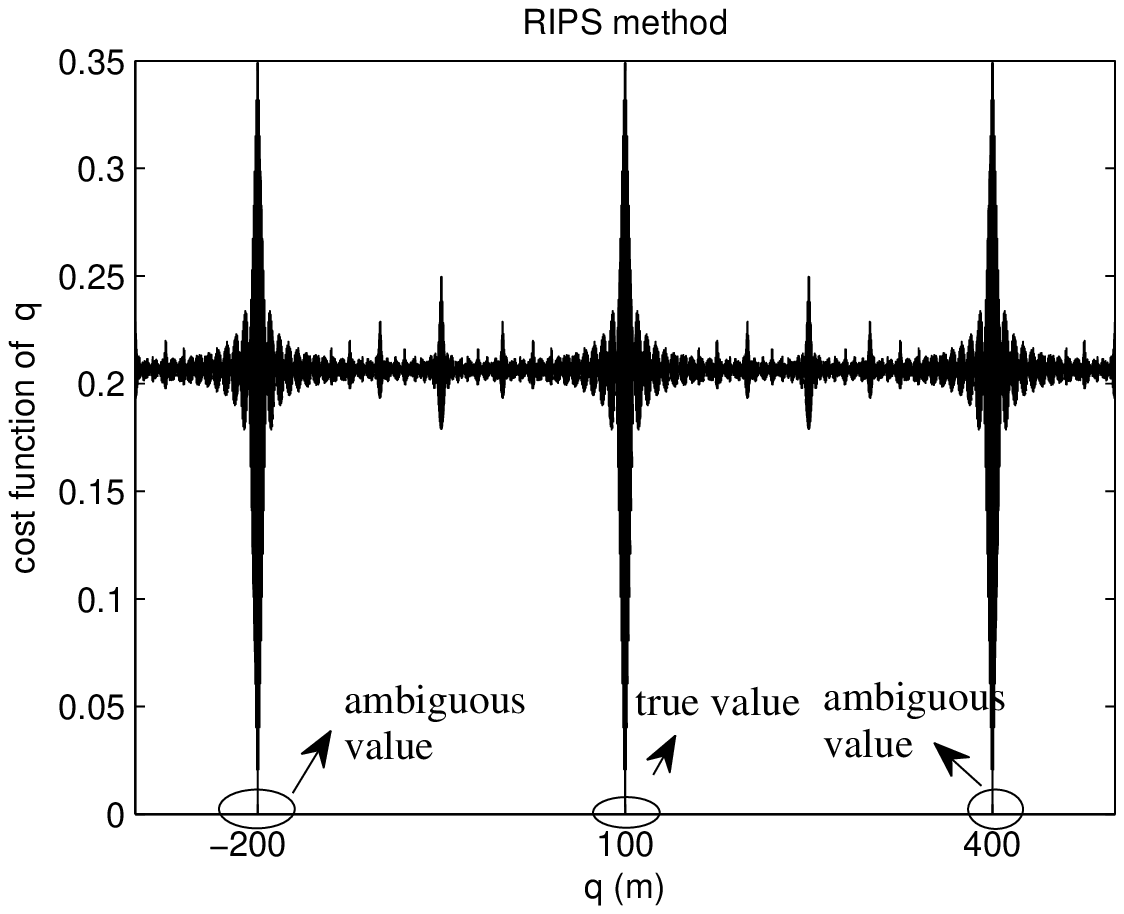}%
         \label{fig_ambiguity_RIPS}}
         \subfloat[Case II]{\includegraphics[height=2.8in,width=3in]{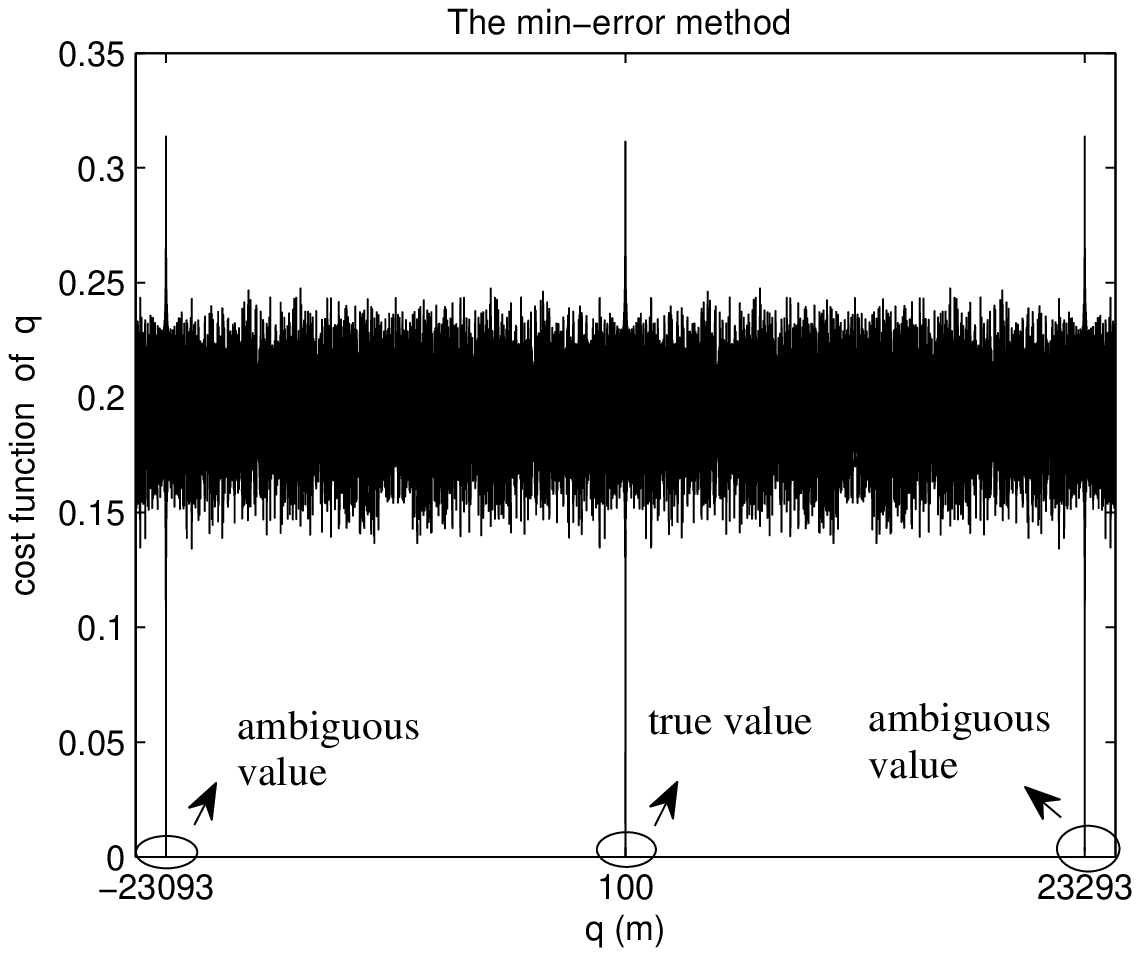}%
         \label{fig_ambiguity_optimal}}
         }
         \hfill               
      \caption{The unambiguous measurement range versus frequency pattern, $q_0=\,100m$.(a) RIPS method.(b) the min-error method.}
      \label{fig_ambiguity}
    \end{figure}

   When SNR is below the conventional threshold and the outlier occurs, the error is uniformly distributed across the entire measurement range and the MSE will not reflect the estimation performance any more (one large outlier may ruin the MSE curve and lead to larger MSE than that caused by many small outliers, so the prime-based min-error method may be inferior to RIPS method in MSE performance for low SNR, see Fig.\ref{fig_double_mse}). Hence, we compare the probability of incorrect unwrapping of different methods rather than MSE in Fig.\ref{fig_false_unwrapping}. From (\ref{equa:basic}) and (\ref{equa:LS_estimation}), the incorrect unwrapping probability $P_f$ is defined as (for correct unwrapping, the error must be less than one wavelength)
    \begin{gather*}
       P_f=P( |\hat q-q_0|>\lambda_N), \quad \lambda_1>\lambda_2\cdots >\lambda_N
    \end{gather*}
   where $\hat q$ is the estimation of $q_0$, the results is averaged over 500,000 Monte Carlo runs. Fig.\ref{fig_false_unwrapping} shows that the proposed min-error method has a much better performance, in term of $P_f$, than both the RIPS method and the max-error method. The min-error method provides approximately 2 dB and 4 dB gain over the RIPS method and the max-error method respectively, at an incorrect unwrapping probability of $10^{-4}$ for $N=21,B=20\,\textrm{MHz}$. Similar observation also holds for $N=41,B=40\,\textrm{MHz}$.

     \begin{figure}[!t]  
      \centerline{
         \subfloat[Case I]{\includegraphics[height=2.8in,width=3in]{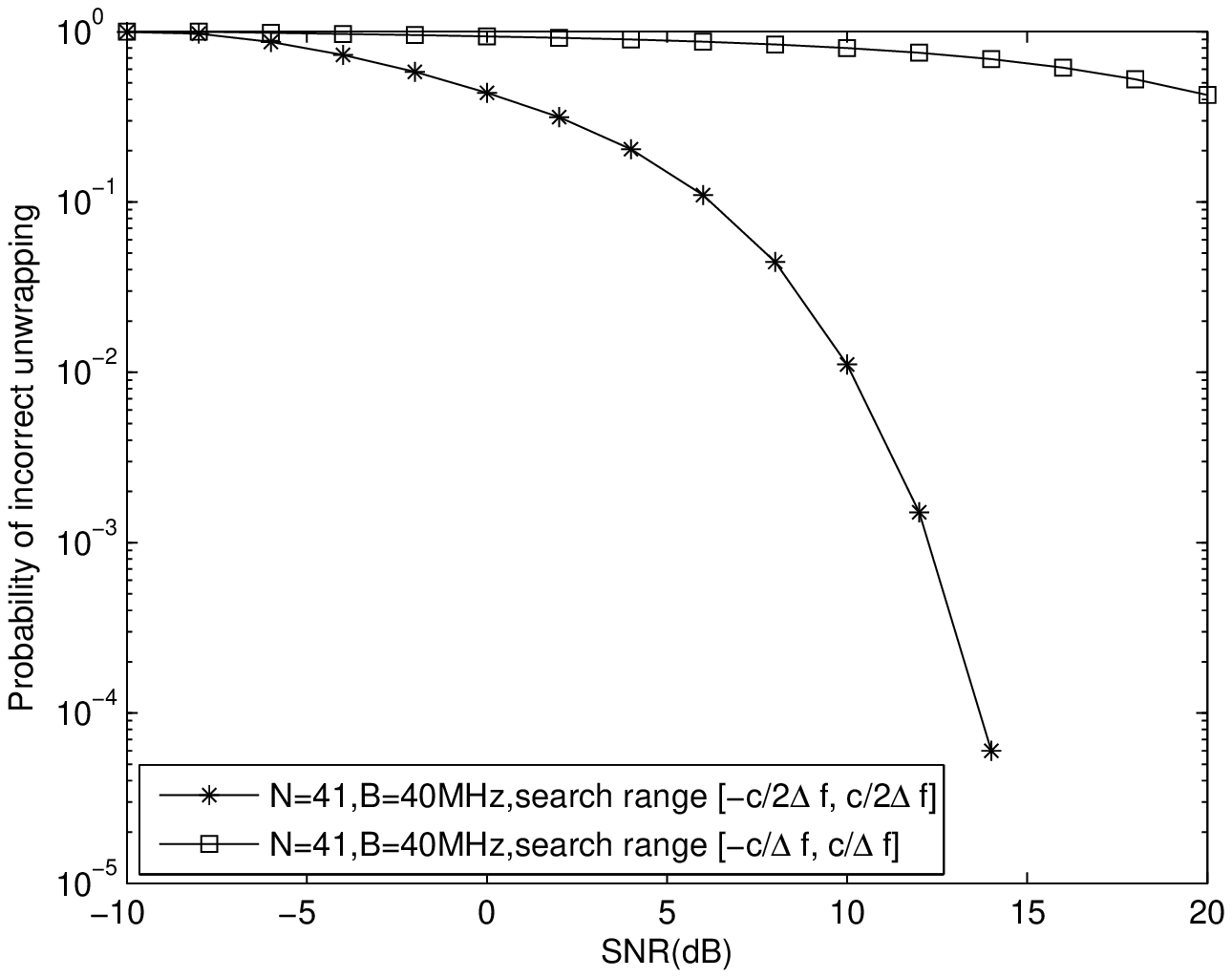}%
         \label{fig_Pe_narrow}}
         \subfloat[Case II]{\includegraphics[height=2.8in,width=3in]{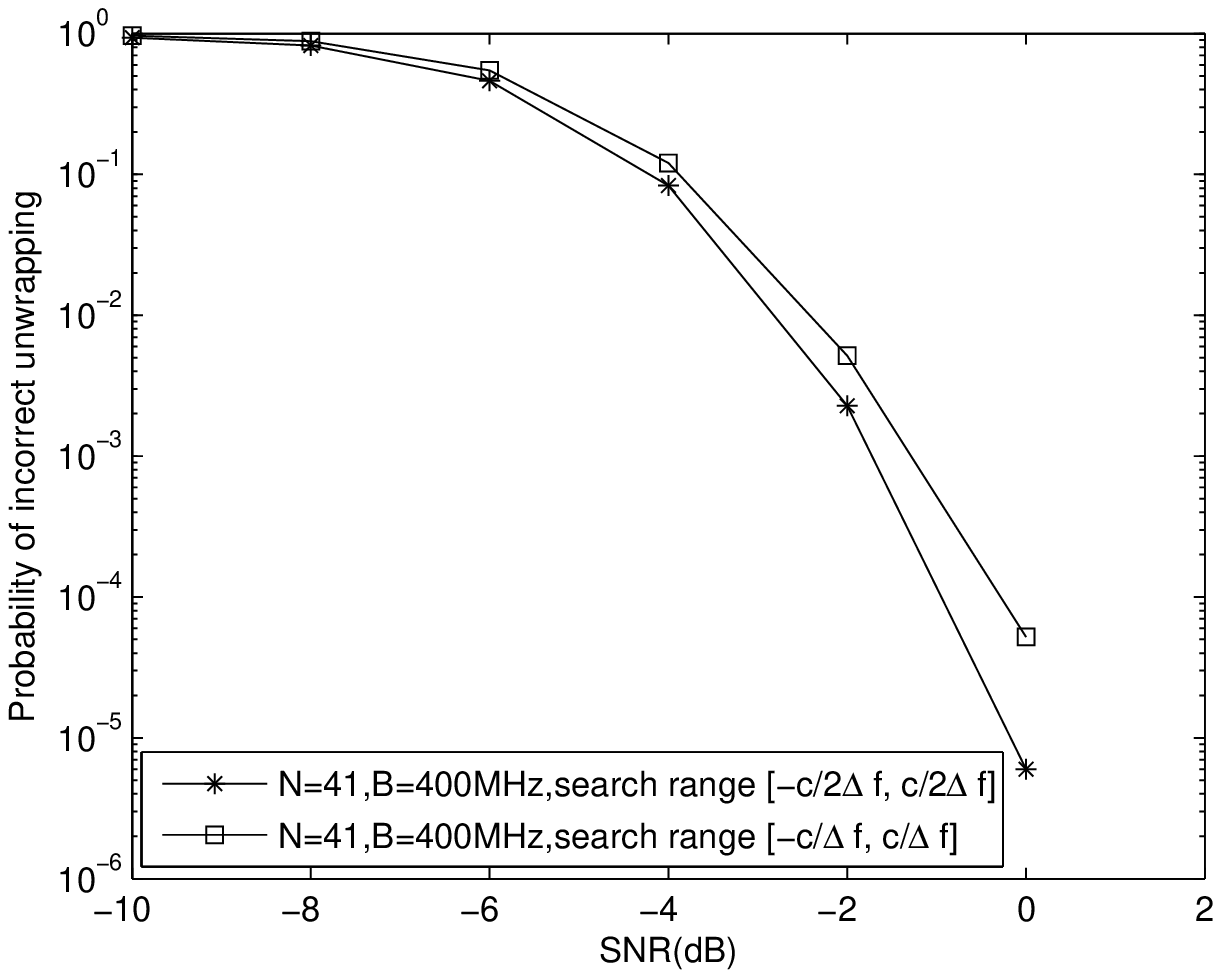}%
         \label{fig_Pe_wide}}
         }
         \hfill               
      \caption{The probability of incorrect phase unwrapping versus B and the search range, RIPS method, $q_0=\,10m$. (a) $N=41$, $B=40\,\textrm{MHz}$, $\Delta f=1 \textrm{MHz}$, $\varepsilon=0.1$ and $f_1=400.1\,\textrm{MHz}$. (b) $N=41$, $B=400\,\textrm{MHz}$, $\Delta f=10 \textrm{MHz}$, $\varepsilon=0.5$ and $f_1=105\,\textrm{MHz}$.}
      \label{fig_UMR_bandwidth}
    \end{figure}

   Fig.\ref{fig_ambiguity} further compares the UMR of the min-error method with RIPS method for the above parameter $N=\nobreak41, B=\nobreak40\,\textrm{MHz}$. The design parameter of the proposed method is $(40,41,65,1,199)$. So the UMR of RIPS method and the proposed method are $\Delta L=c/(B/(N-1))=300\,\textrm{m}$ and $\Delta L=c/(K\Delta f_{\min})- \varepsilon{\sum_{i=1}^N \lambda _i ^{-1}}\big/{\sum_{i=1}^N \lambda_i^{-2} }\approx23193\,\textrm{m}$, which both are in good agreement with simulation results. The proposed method achieves far more large UMR than RIPS. From Fig.\ref{fig_ambiguity_optimal}, it is also noted that larger sidelobes are not seen for being uniformly averaged over the whole parameter space and sharp peak is observed at the true location. These properties result in superior estimation accuracy and low outlier probability verified by the MSE curve. Moreover, the UMR can be easily enlarged by adjusting the parameter $K$.

    The incorrect unwrapping probability of RIPS method under different bandwidth and initial measurement frequency is shown in Fig.\ref{fig_UMR_bandwidth}. The search range of both $[-c/2\Delta f, c/2\Delta f]$ and $[-c/\Delta f, c/\Delta f]$ are evaluated. It is clear that when $f_1/B\gg 1$ or $\varepsilon\rightarrow 0$, the measurable range is constrained by P-UMR. Searching outside this range will obtain an ambiguous value and result in a large incorrect probability. For example, an error probability as large as $70\%$ is observed for $\textrm{SNR} =10\textrm{dB}$ when two ambiguous solutions are included into the search range, while the error probability of only $1\%$ is obtained for the same SNR if the P-UMR constraint is considered. However, the conclusion is invalid for $f_1/B< 1$, as is seen in Fig.\ref{fig_Pe_wide}. In the latter case, P-UMR becomes a pessimistic estimation of range since the measurable range is much more large than it. The observations reveal that the UMR is highly correlated with $f_1/B$ but less affected by noise variance. This result agrees well with the analysis in section \ref{subsect:UMR}.

\section{Experimental Evaluation}
\label{sect:experiment}

   The low-cost mica2 nodes is exploited for field experiment. All the ranging procedure is similar to the one in \cite{Maroti:05} and \cite{Kusy:06} except the measurement frequencies. Five nodes are used with two transmitters (A and B) and three receivers (C$\sim$E). For each measurement round, two transmitters and two receivers are needed. So the nodes form three deployment scenarios, ie.ABCD, ABCE and ABDE. The nodes are deployed in football-field of our campus, see Fig.\ref{fig_deployment}.
   \begin{figure}[!t]
   \centering
   \includegraphics[height=3in,width=2in]{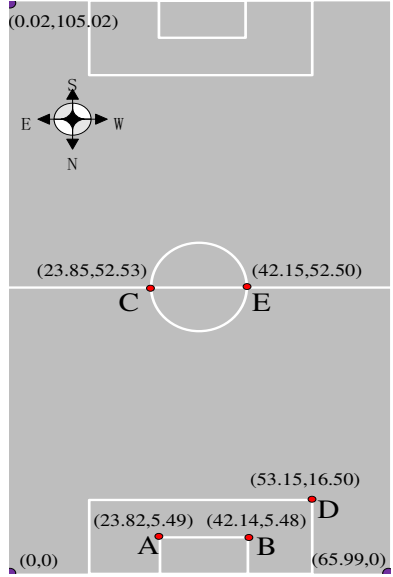}
   \caption{Coordinates of nodes (m).}
   \label{fig_deployment}
   \end{figure}
   To alleviate the multipath effect, all the nodes are placed one-meter above the ground, seen \cite{Kusy:06} for details. The real coordinates of nodes are determined via differential GPS, which has position error of about $\pm 2\,\textrm{cm}$. The true ranges are $d_{ABCD}=19.19\,\textrm{m}$, $d_{ABCE}=6.88\,\textrm{m}$ and $d_{ABDE}=-12.31\,\textrm{m}$, where $d_{ABCD}$ is the linear combination of the distances in the form of  $d_{ABCD} = d_{AD}-d_{BD}+d_{BC}-d_{AC}$, $d_{XY}$ denotes the distance between node $X$ and $Y$.

   The measurement frequencies range from 410\,MHz to 450.378\,MHz, the bandwidth and frequency number are $B=\nobreak40.378\,\textrm{MHz}$, $N=31$. The following methods are compared under identical $B$ and $N$: RIPS, min-error method, max-error method and the random method. For random method, the measurement frequencies are picked randomly from the usable frequency band. The experiment parameter of prime-based min-error method is $(B,N,\Delta f_{\min},i,K)=(40.378,31,65,12,200)$. Since $i=\nobreak12$, then the prime sequence described in section \ref{subsect:algorithm} is $\mathcal S=[37\, 41\, 43 \,47\, 53\, 59\, 61\, 67\, 71\, 73\, 79\, 83\, 89\, 97\, 101\, 103\, 107\, 109\,\\ 113\, 127\, 131\, 137\, 139\, 149\, 151\, 157\, 163\, 167\, 173\, 179]$. The nodes have fine frequency resolution of $\Delta f_{\min}=65\,\textrm{Hz}$. It is easy to verify that $\sum_i {S\left(i\right)} K\Delta f_{\min }=B$ holds. Based on Theorem \ref{theo:UMR_modify}, the practical UMR is $\Delta L = c/ \left( K\Delta f_{\min } \right)- \varepsilon{\sum_{i=1}^N \lambda _i ^{-1}}\big/{\sum_{i=1}^N \lambda_i^{-2} }\approx 23.077 \textrm{km}$. For each method and deployment, 50 independent experiments are performed.

   Since the UMR is $\Delta L\approx23.077\,\textrm{km}$, the search ranging is set to [-1000\,m,24000\,m] to guarantee that one ambiguity solution could be searched in Fig.\ref{fig_UMR_experiment}. The error between the real and estimated range is then plotted in Fig.\ref{fig_UMR_experiment}.

    \begin{figure}[!t]  
      \centerline{
         \subfloat[Case I]{\includegraphics[height=2.8in,width=3in]{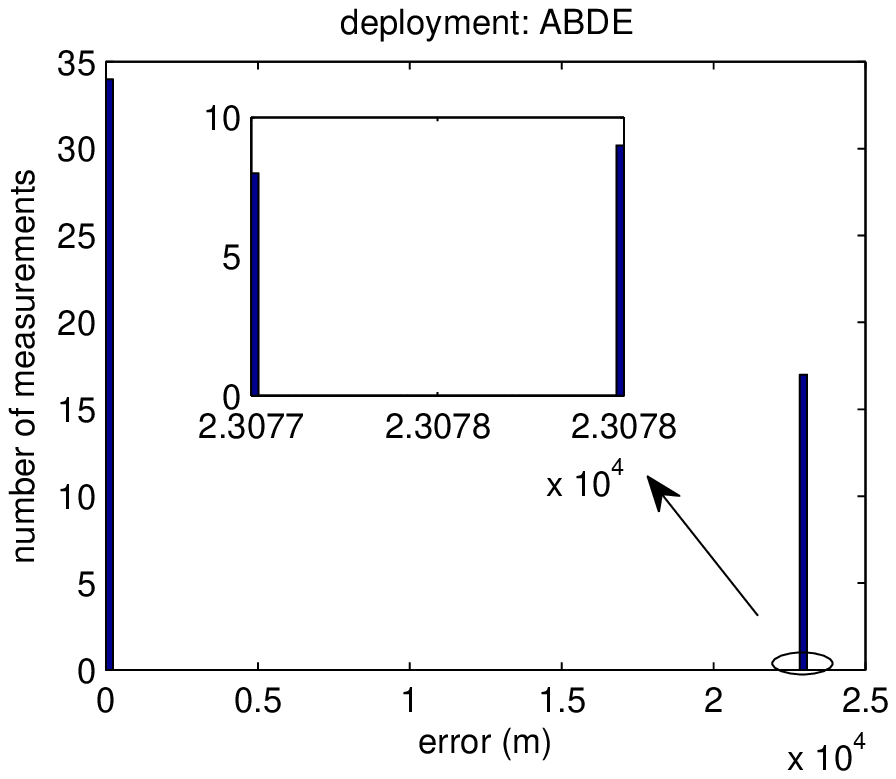}%
            \label{fig_UMR_min}}
         \subfloat[Case II]{\includegraphics[height=2.8in,width=3in]{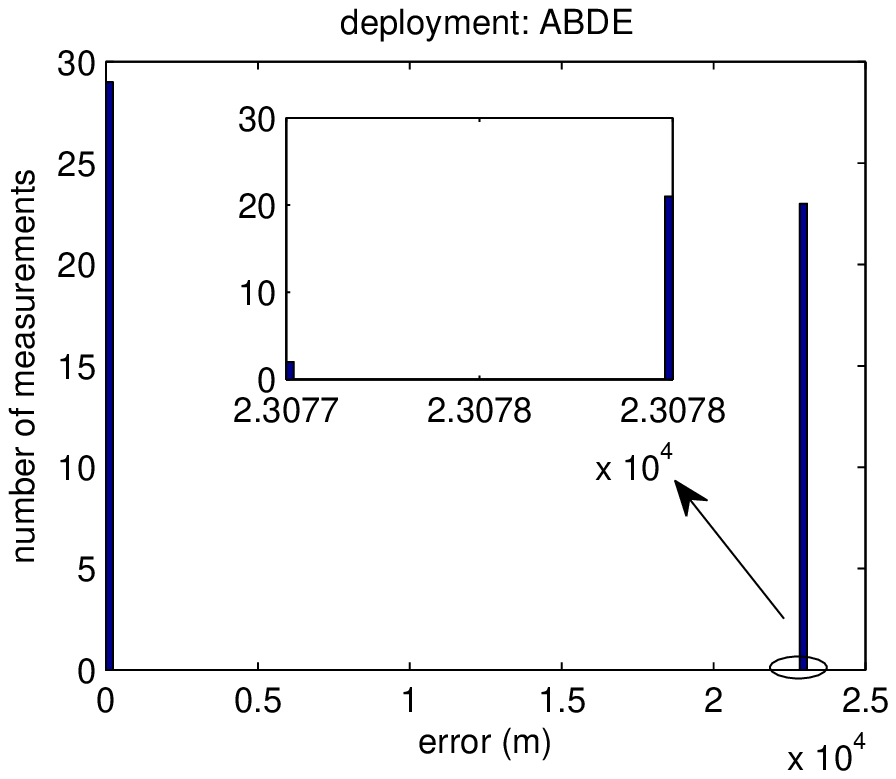}%
            \label{fig_UMR_max}}
            }
         \hfill               
      \caption{Validation of the unambiguous measurement range via field experiment. (a) the min-error method. (b) the max-error method.}
      \label{fig_UMR_experiment}
    \end{figure}

   It is obvious that the errors lie in two regions, one is near zero and the other is located at 23077\,m or 23078\,m. The latter is very close to the practical UMR(P-UMR). More importantly, the estimation failure probability(the error is equal to the P-UMR) is as large as $32\%$ and even $44\%$ for the min-error and max-error method respectively. This implies that the P-UMR does exist and is robust to noise and frequency inaccuracy, which are inevitable especially for low-cost hardware such as mica2 node.

     \begin{figure}[!t]  %
      \centerline{
         \subfloat[Case I]{\includegraphics[height=2in,width=2in]{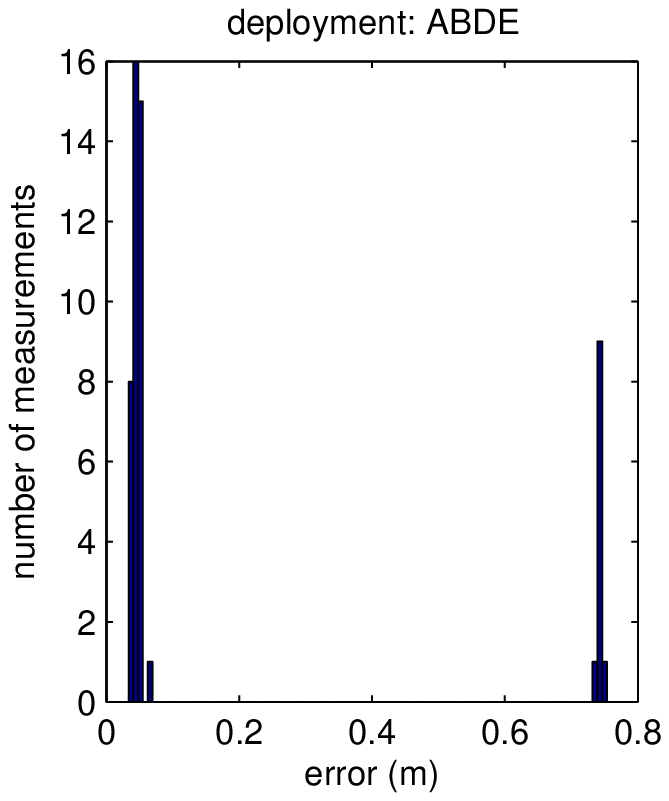}%
            \label{fig_min_error}}
         \subfloat[Case II]{\includegraphics[height=2in,width=2in]{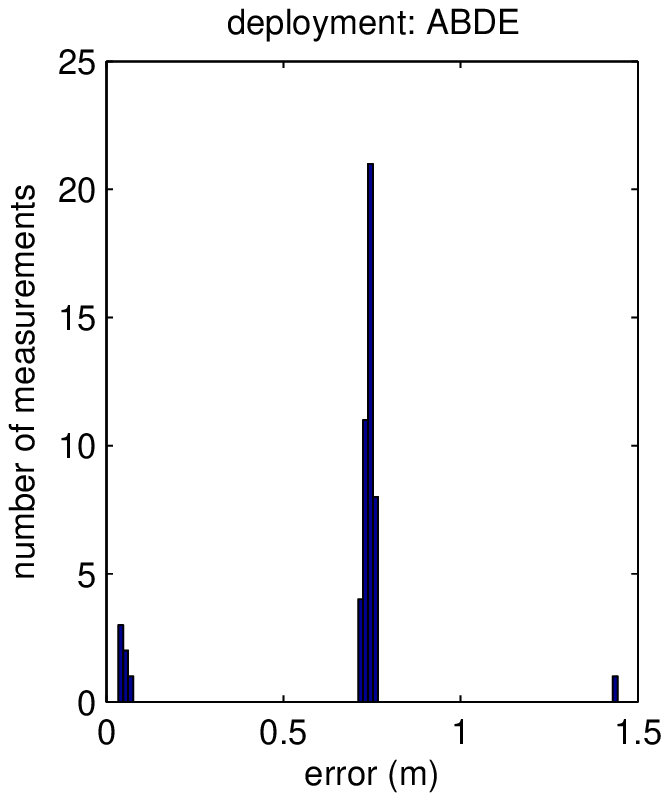}%
            \label{fig_max_error}}
         }
         \hfill               
         \vspace{-4mm}
      \centerline{
         \subfloat[Case II]{\includegraphics[height=2in,width=2in]{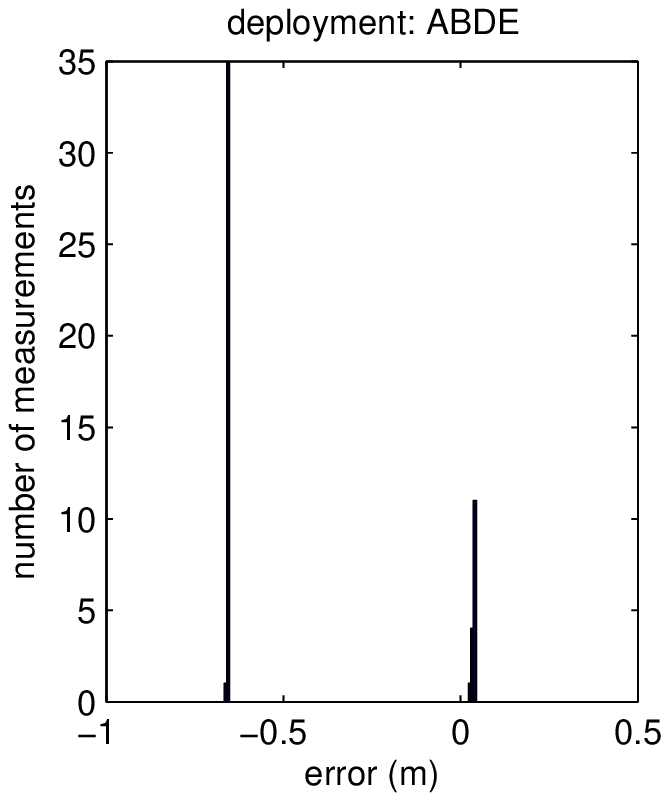}%
           \label{fig_RIPS}}
         \subfloat[Case II]{\includegraphics[height=2in,width=2in]{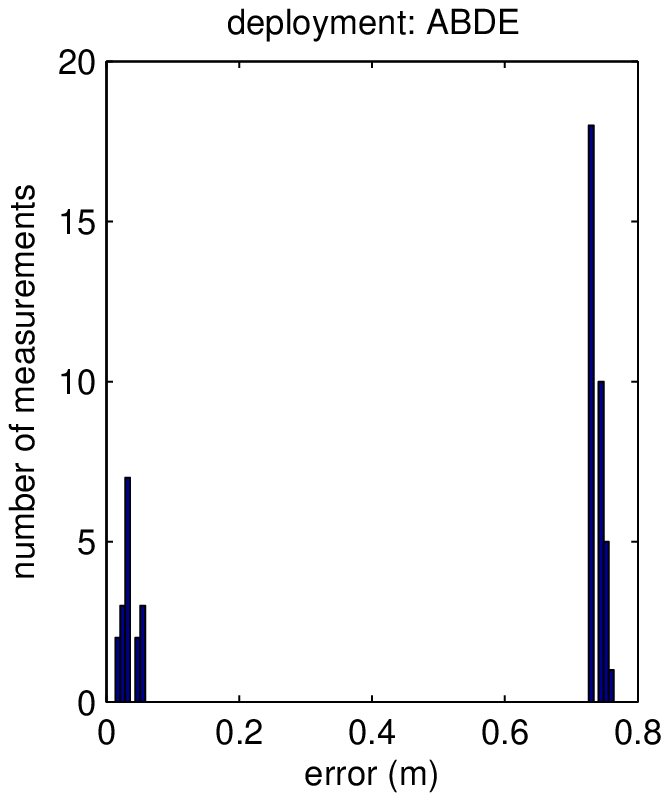}%
           \label{fig_random1}}
         }
      \caption{Error distribution of range estimation with 50 independent experiments for each method. (a) min-error method.
      (b) max-error method. (c) RIPS method. (d) random method. }
      \label{fig_Error_distribution}
    \end{figure}

    \begin{figure}[!t]
     \centering
     \includegraphics[height=2.8in,width=3.5in]{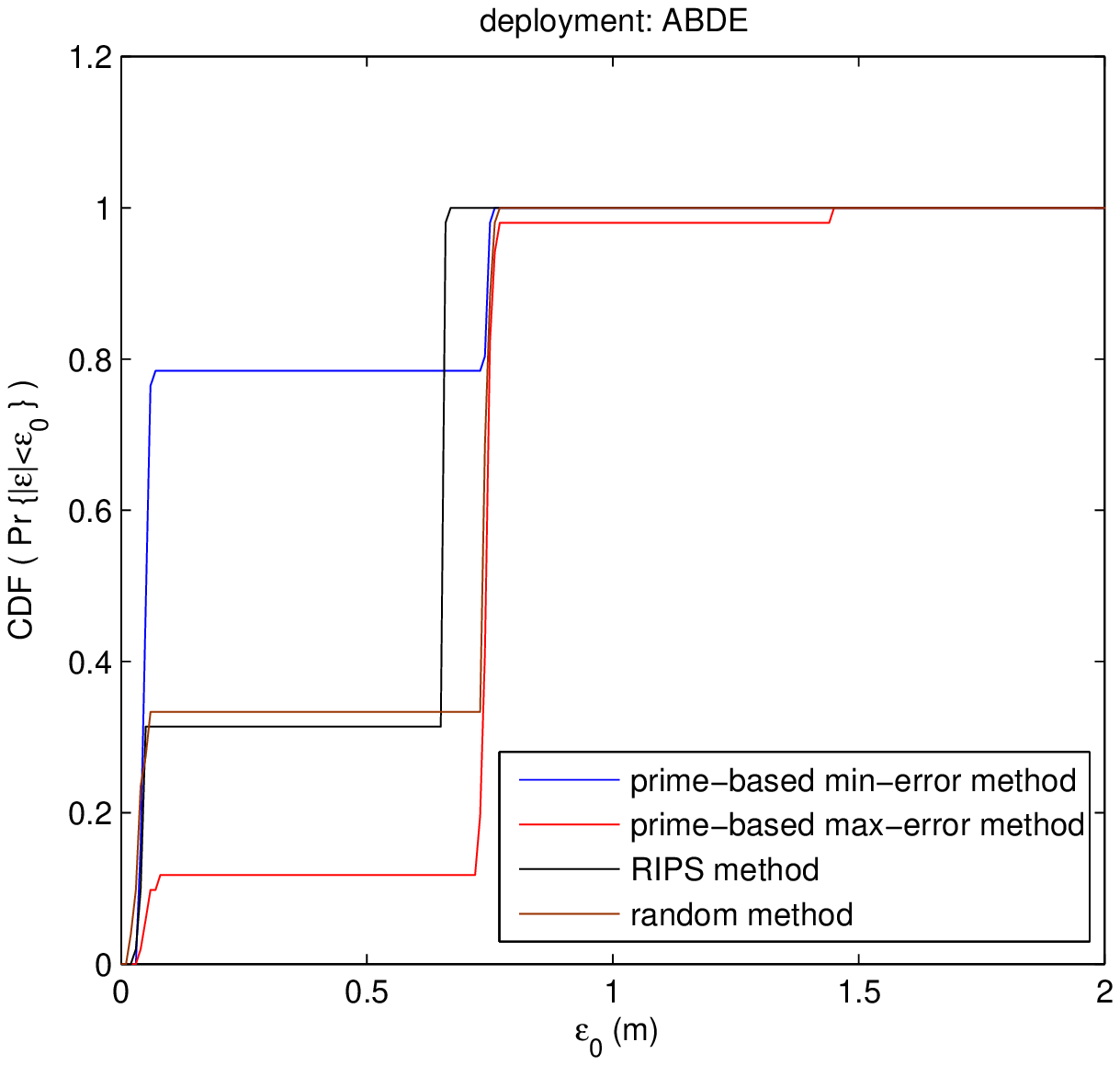}
     \caption{CDF of the absolute errors $|\epsilon|$ for different methods.}
     \label{fig_CDF}
   \end{figure}

   The search operation of the following experiment is performed inside the UMR range of RIPS, ie. $[-c/2\Delta f,c/2\Delta f]\approx[-100\,\textrm{m},100\,\textrm{m}]$ to avoid ambiguity. The distribution of ranging error is shown in Fig.\ref{fig_Error_distribution} for different methods. Fig.\ref{fig_CDF} exhibits the cumulative distribution function (CDF) of the absolute value of error for the deployment of ABDE. These two figures confirm the results in Fig.\ref{fig_false_unwrapping}. It is pointed out that the ranging errors are bias and not Gaussian distribution any more due to the existence of multipath or wrapped Gaussian noise, as is seen in Fig.\ref{fig_Error_distribution}. Even in this case, the proposed min-error method is still superior to the other schemes.

   The MSE performance of different deployments are also shown in Fig.\ref{fig_MSE}. The discrepancies are observed for all the methods except the min-error method, which has the best accuracy all the time. This discrepancies can be attributed to the multipath effect. It is well known that frequency-selective shading is introduced in multipath channel and it varies from one location to another. This is so-called frequency-selective and space-selective property of wireless channel. Therefore, the measurement frequencies of a certain method may undergo deep shading in one deployment and leading to small SNR in the receiver. It is quite possible that the opposite happens for another deployment. Consequently, the discrepancies appear.

   \begin{figure}[!t]
     \centering
     \includegraphics[height=2.8in,width=3.5in]{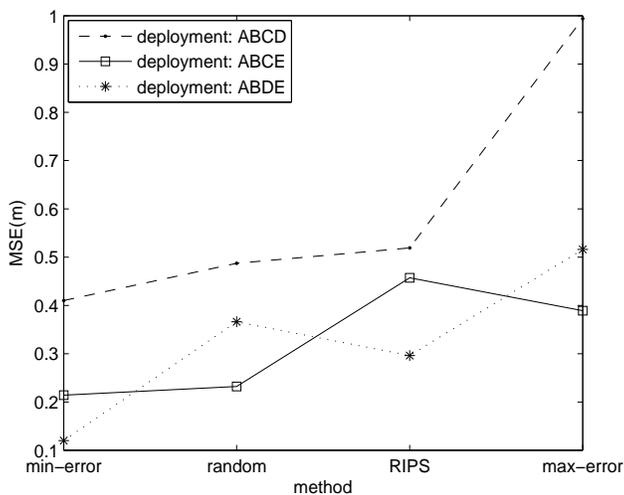}
     \caption{MSE for different methods with three deployments, averaged over 50 independent experiments for each method and
             deployment.}
     \label{fig_MSE}
   \end{figure}

\section{Conclusion and Discussion}
   In this paper, we focused on the frequency optimization of MFI system to extend the UMR, decrease the probability of outlier and improve the estimation accuracy. The main contributions of the paper can be summarized as follows
   \begin{itemize}
     \item  We prove that the UMR of MFI is in inverse proportion to the greatest common divisor (GCD) of frequency interval, thus the measurement range may be simply extended by reducing the GCD. Furthermore, we point out that the theoretical UMR of MFI is too optimistic for narrowband ranging ($f_1/B\gg 1$), while the P-UMR may be conservative when $f_1/B< 1$. This point has not been mentioned in the literature before.
     \item We explore the relationship between outlier probability and the ambiguity function as well as the probability density function, and suggest to use prime-based frequency interval for outlier suppression.
     \item The unique ``double threshold" phenomenon of MFI is firstly discovered and the expressions for the MMSE and HMSE are derived.
     \item Focusing on the optimization of MMSE performance, we present an optimal permutation for any set of frequency interval and prove its optimality for the first time.
     \item Based on the finding mentioned above, we present a quite simple and effective frequency design method and simulation results verified that the proposed method outperforms the existing method in UMR and MSE simultaneously.
     \item Different from most of the related work, field experiments have been designed and further demonstrate the new method's robustness to practical interference such as frequency inaccuracy and multipath error.
   \end{itemize}

     Although the theory of the paper is developed under the assumption of relatively narrowband ranging with $f_1\gg B$, it is found that the proposed method also performs well in relatively wide bandwidth. As a final remark, since the outlier suppression by non-searching-based frequency interval design is an open problem in the literature, we just give a heuristic design method and further research is required in this area.



\appendices    
\section{Proof of the Theorem \ref{theo:UMR_modify} }
\label{append:UMR}
\renewcommand{\labelenumi}{(\arabic{enumi})}
\begin{enumerate}
   \item
   Suppose the true range is $q_0$. Since $f_i=(k_i+\varepsilon)\Delta f_{\min}$, it is obvious that the cost function could not achieve zero at $q_0+\Delta L$, $\Delta L=c/\Delta f_{\min}$. We will instead consider the location of local minimal $q =q_0 +\Delta L+x$, near the original ambiguous location. The cost function is rewritten as
   \begin{align}
      \label{equa:cost_general_no_noise}
      S( q )& =\sum_{i = 1}^N \left( \left[\varphi (i) - \hat \varphi_q (i) \right]_{2\pi} \right)^2 \notag\\
      &= \sum_{i=1}^N \left(\left[2\pi\frac{q_0}{\lambda_i} - 2\pi \frac{(q_0+\Delta L+x )}{\lambda_i} \right]_{2\pi} \right)^2  \notag\\
      & = \sum_{i = 1}^N \left( \left[\frac{2\pi x}{\lambda_i} + \frac{2\pi\left((k_i+\varepsilon )\Delta f_{\min} \right)\Delta L}{c}  \right]_{2\pi}  \right)^2 \notag \\
      &=\sum_{i=1}^N \left(\left[ \frac{2\pi (x+\varepsilon \lambda_i)}{\lambda_i} \right]_{2\pi}  \right)^2\notag \\
      &=T(x)
   \end{align}

   Assume $\lambda_1>\lambda_2\cdots >\lambda_N$ and $\lambda_1<2\lambda_i$ (it holds for $f_1\geq B$), and note that $|x|<|\varepsilon|\lambda_1$ and $sign(x)=-sign(\varepsilon)$ for $x=-\varepsilon {\sum_{i=1}^N \lambda _i ^{-1}}/ {\sum_{i=1}^N \lambda_i^{-2} }$, then $| x+\varepsilon \lambda_i| \leq max\{|\varepsilon\lambda_i|,|\varepsilon(\lambda_1-\lambda_i)| \}< \lambda_i/2$.
    We have
    \begin{align}
      \label{equa:cost_PUMR_no_noise}
      T( x|x=-\varepsilon \frac{\sum_{i=1}^N \lambda _i ^{-1}}{\sum_{i=1}^N \lambda_i^{-2} } )&= \sum_{i=1}^N \left(\frac{2\pi ( x+\varepsilon \lambda_i )}{\lambda_i} \right)^2 \notag\\
      &=\sum_{i=1}^N  \left( 2\pi \varepsilon \left( \lambda_i - \frac{\sum_{k=1}^N \lambda _k ^{-1}} {\sum_{k=1}^N \lambda_k^{-2} }\right)/\lambda _i  \right)^2
    \end{align}

   Since
   \begin{gather}
      \sum_{k=1}^N {\lambda_k^{-1}}\bigg/\sum_{k=1}^N {\lambda_k^{-2}}
            < \sum_{k= 1}^N  \left(\lambda_k ^{-1} \lambda_1 \lambda_k ^{-1}  \right) \bigg/\sum_{k=1}^N {\lambda_k^{-2}}
            = \lambda_1 \notag\\
      \sum_{k=1}^N {\lambda_k^{-1}}\bigg/\sum_{k=1}^N {\lambda_k^{-2}}
            > \sum_{k= 1}^N  \left(\lambda_k ^{-1} \lambda_N \lambda_k ^{-1}  \right) \bigg/\sum_{k=1}^N {\lambda_k^{-2}}
            = \lambda_N
    \end{gather}
   We obtain
   \begin{gather}
       T(x) < \sum_{i=1}^N \left( 2\pi \varepsilon ( \lambda_1 - \lambda_N)/\lambda_N \right)^2
                = 4N\pi^2 \varepsilon^2 \left(f_1/B \right)^{-2}
   \end{gather}
   where $B=f_N-f_1$. we finally obtain
   \begin{gather}
      \underset{f_1/B\rightarrow\infty}{\lim} S(q|q=q_0 +\Delta L+x) = 0 \notag\\
      \underset{\varepsilon\rightarrow 0}{\lim}  \quad S(q|q=q_0 +\Delta L+x) = 0
   \end{gather}

   We also want to find the local minimal $x$ of $T(x)$ within the constraint $0<x<-\varepsilon\lambda_1$ for $\varepsilon<0$, or $-\varepsilon\lambda_1<x<0$ for $\varepsilon>0$, we obtain

    \begin{gather}
       \hat x=\underset {x}{\text{argmin}}  \sum_{i=1}^N \left(\frac{2\pi (x + \varepsilon \lambda_i )}{\lambda_i} \right)^2
            = -\varepsilon \frac{\sum_{i=1}^N \lambda _i ^{-1}} {\sum_{i=1}^N \lambda_i^{-2} }
    \end{gather}
   That is to say, $q=q_0 +c/\Delta f_{\min}- \varepsilon{\sum_{i=1}^N \lambda_i ^{-1}}\bigg/{\sum_{i=1}^N \lambda_i^{-2}}$ is one of the local minimal of $S( q )$.

   \item Suppose there exists at least one $\Delta l$ satisfying $\underset{f_1/B\rightarrow\infty}{\lim}
             S(q|q=q_0 +\Delta l)=0$ and $0<\Delta l<c/\Delta f_{\min}- \varepsilon{\sum_{i=1}^N \lambda_i^{-1}}\bigg/{\sum_{i=1}^N \lambda_i^{-2}}$.
             Let $\Delta {l}^{'}$ be the minimal one, then we have
             $$c/\Delta f_{\min}- \varepsilon{\sum_{i=1}^N \lambda _i ^{-1}}\big/{\sum_{i=1}^N \lambda_i^{-2} }=n\Delta {l}^{'}$$
             where $n$ is integer. Then $\Delta {l}^{'}=c/(n\Delta f_{\min})-(\varepsilon/n){\sum_{i=1}^N \lambda _i ^{-1}}\big/{\sum_{i=1}^N \lambda_i^{-2} }$, and
           \begin{align}
                \hspace{-4pt}&S(q|q=q_0 +\Delta {l}^{'})\notag \\ 
                             &=\sum_{i = 1}^N \left( \left[\frac{2\pi\left(k_i+\varepsilon \right)\Delta f_{\min}\Delta{l}^{'}}{c}  \right]_{2\pi}  \right)^2 \notag \\
                             &=\sum_{i = 1}^N \left(\left[2\pi\left(\frac{k_i}{n}+\frac{\varepsilon}{n\lambda_i}\left(\lambda_i-\frac{\sum_{k=1}^N \lambda _k ^{-1}} {\sum_{k=1}^N \lambda_k^{-2} } \right)  \right)\right]_{2\pi}  \right)^2 \notag \\
                             &=(2\pi)^2 \sum_{i = 1}^N\left(\left[\frac{k_i}{n}+\frac{\varepsilon}{n\lambda_i}\left(\lambda_i-\frac{\sum_{k=1}^N \lambda _k ^{-1}}
                                {\sum_{k=1}^N \lambda_k^{-2}} \right) \right]_{1} \right)^2
          \end{align}
          Since the GCD is $\Delta f_{\min}$, there at least exists one $j$,$1\leq j\leq N$, $[k_j]_{n}\neq 0$, where $[.]_{n}$ is modulo $n$ operation, yielding a result in the interval $[-n/2, n/2)$. Assume $[k_j]_{n}=p$, with the integer $p$, $-\lfloor n/2\rfloor\leq p<0$ or $0<p\leq\lceil n/2\rceil-1$, $\lfloor .\rfloor$, $\lceil .\rceil$ denote the floor and ceiling funcion, then,
          \begin{align}
                 \hspace{-4pt}&S(q|q=q_0 +\Delta {l}^{'})\notag\\
                                          &> (2\pi)^2\left(\left[\frac{k_j}{n}+\frac{\varepsilon}{n\lambda_j}\left(\lambda_j-\frac{\sum_{k=1}^N \lambda_k ^{-1}}
                                             {\sum_{k=1}^N \lambda_k^{-2} } \right) \right]_{1}  \right)^2 \notag\\
                                          &= (2\pi)^2 \left( \left[\frac{p}{n}+\frac{\varepsilon}{n}\left(\lambda_j-\frac{\sum_{k=1}^N \lambda_k ^{-1}} {\sum_{k=1}^N \lambda_k^{-2} } \right)\Bigg/\lambda_j \right]_{1}  \right)^2
          \end{align}
          For $f_1>B$, $|\left(\lambda_j-{\sum_{k=1}^N \lambda_k ^{-1}}\big/ {\sum_{k=1}^N \lambda_k^{-2} } \right)\Big/\lambda_j|<B/f_1<1$. If
            \begin{enumerate}
              \renewcommand{\baselinestretch}{1.2} \normalsize
              \item $n$ is odd, or $n$ is even and $p\neq-n/2$, then
              \begin{align*}
                 \hspace{-4pt}&S(q|q=q_0 +\Delta {l}^{'}) \\
                                          &> (2\pi)^2 \left(\frac{p}{n}+\frac{\varepsilon}{n}\left(\lambda_j-\frac{\sum_{k=1}^N \lambda_k ^{-1}} {\sum_{k=1}^N \lambda_k^{-2} } \right)\Bigg/\lambda_j  \right)^2  \\
                                          &> (2\pi)^2 \left(\frac{p^2}{n^2}+2\frac{p\varepsilon}{n^2}\left(\lambda_j-\frac{\sum_{k=1}^N \lambda_k ^{-1}} {\sum_{k=1}^N \lambda_k^{-2} } \right)\Bigg/\lambda_j  \right)\\
                                          &> (2\pi)^2 \left(\frac{p^2}{n^2}-2\frac{|p\varepsilon|}{n^2}\frac{B}{f_1} \right)
              \end{align*}
              \vspace{-10pt}
              \begin{gather*}
              \hspace{-30pt}\underset{f_1/B\rightarrow\infty}{\lim}
                 S(q|q=q_0 +\Delta {l}^{'})> (2\pi)^2 \frac{p^2}{n^2}
              \end{gather*}

              \item $n$ is even, and $p=-n/2$, it is also easily to verify
              \begin{gather*}
              \hspace{-30pt}\underset{f_1/B\rightarrow\infty}{\lim}
                 S(q|q=q_0 +\Delta {l}^{'})> (2\pi)^2 \left(\frac{n-1}{2n} \right)^2
              \end{gather*}
            \end{enumerate}
          Note that  $\underset{f_1/B\rightarrow\infty}{\lim}
                 S(q|q=q_0 +\Delta {l}^{'})=0$ according to the assumption.
          Contradiction. This completes the proof.

    \item   Similar to (\ref{equa:cost_general_no_noise}) and  (\ref{equa:cost_PUMR_no_noise}), we have
          \begin{align}
            \label{}
            S\left(q|q=q_0+c/\Delta f_{\min}\hspace{-2pt}-\hspace{-2pt}\varepsilon{\sum_{i=1}^N \lambda_i^{-1}}\bigg/{\sum_{i=1}^N \lambda_i^{-2}}\right)
            &=\sum_{i=1}^N \left(\left[2\pi\varepsilon\left(\lambda_i- \frac{\sum_{k=1}^N \lambda _k^{-1}} {\sum_{k=1}^N \lambda_k^{-2}}\right)/\lambda _i+\theta_e(i) \right]_{2\pi} \right)^2 \notag\\
            S\left(q|q=q_0\right)
            &=\sum_{i=1}^N \left(\left[\theta_e(i) \right]_{2\pi} \right)^2
          \end{align}
      Note that $\left|\left(\lambda_i- \frac{\sum_{k=1}^N \lambda _k^{-1}} {\sum_{k=1}^N \lambda_k^{-2}}\right)/\lambda_i\right|<\frac{B}{f_1}$. According to the $3\sigma$ principle with $\textrm{SNR}=1/2\sigma^2>0\textrm{dB}$ and $f_1/B\geq4$, the following equation holds with $99.7\%$ probability,
           \begin{gather*}
              \left| \theta_e(i)\right|<3\sigma<\frac{3\sqrt{2}}{2\pi}\pi<0.7\pi
           \end{gather*}
          \begin{align}
            \label{}
            S\left(q|q=q_0+c/\Delta f_{\min}\hspace{-2pt}-\hspace{-2pt}\varepsilon{\sum_{i=1}^N \lambda_i^{-1}}\bigg/{\sum_{i=1}^N \lambda_i^{-2}}\right)
               &=\sum_{i=1}^N \left(2\pi\varepsilon\left(\lambda_i- \frac{\sum_{k=1}^N \lambda _k^{-1}} {\sum_{k=1}^N \lambda_k^{-2}}\right)/\lambda _i+\theta_e(i)  \right)^2 \notag\\
           S\left(q|q=q_0\right)
               &=\sum_{i=1}^N \left(\theta_e(i) \right)^2
          \end{align}
      Denote $w_i=2\pi\varepsilon\left(\lambda_i- \frac{\sum_{k=1}^N \lambda _k^{-1}} {\sum_{k=1}^N \lambda_k^{-2}}\right)/\lambda_i$, then
      \begin{align}
                P_a&=P\left( S\left(q_0+c/\Delta f_{\min}\hspace{-2pt}-\hspace{-2pt}\varepsilon{\sum_{i=1}^N \lambda_i^{-1}}\bigg/{\sum_{i=1}^N \lambda_i^{-2}}\right)< S\left(q_0\right) \right)   \notag\\
                &=P\left( \sum_{i=1}^N \left(w_i+\theta_e(i)  \right)^2< \sum_{i=1}^N \left(\theta_e(i)\right)^2 \right) \notag\\
                &=P\left( \sum_{i=1}^N {w_i\theta_e(i)} <-\frac{1}{2}\sum_{i=1}^N w_{i}^{2}  \right)
     \end{align}
     Let $V=\sum_{i=1}^N {w_i\theta_e(i)}$, $V$ is then Gaussian random variable with zero mean and variance $E\{ V^{2}\}=\sigma_{V}^{2}=\sigma_{\theta}^{2}\sum_{i=1}^N w_{i}^{2}$, and $W=2\pi|\varepsilon|\frac{B}{f_1}$. We obtain
     \begin{align}
                P_a&=P\left( V <-\frac{\sigma_{V}^{2}}{2\sigma_{\theta}^{2}}  \right)  \notag\\
                   &=\frac{1}{2} \left(1+ erf\left(\frac{-\sigma_{V}}{2\sqrt{2}\sigma_{\theta}^{2}} \right) \right)   \notag\\
                   &=\frac{1}{2} \left(1- erf\left(\frac{\sqrt{\sum_{i=1}^N w_{i}^{2}}}{2\sqrt{2}\sigma_\theta } \right) \right)   \notag\\
                   &>\frac{1}{2} \left(1- erf\left(\frac{\sqrt{N} W}{2\sqrt{2}\sigma_\theta } \right) \right)
     \end{align}
     where $erf\left(x\right)=\frac{2}{\sqrt{\pi}} \int_0^x {e^{-t^2}}dt$.
\end{enumerate}

\section{Proof of the Theorem \ref{theo:optimal_frequency}}
\label{append:optimal_frequency}

    We firstly sort the $N-1$ adjacent frequency spacing in ascending order
    \begin{gather*}
       \Delta \vec{\tilde{f}}=[\Delta f_{1},\Delta f_{2},\cdots\Delta f_{N-1}]=[k_1,k_2,\cdots k_{N-1}]\Delta f_{\min}\\
        \sum_{i=1}^{N-1}k_i=B/\Delta f_{\min}=M\\
        1\leq k_1\leq k_2\leq\cdots k_{N-1},k_{N-1}\leq M+2-N
    \end{gather*}
    According to Theorem \ref{theo:optimal_permutation}, the optimal rearrangement of the spacing is
    \begin{align*}
       \Delta \vec{f}&=[\Delta f_1,\Delta f_3,\cdots\Delta f_{N-1},\Delta f_{N-2},\cdots\Delta f_4,\Delta f_2]\\
                        &=[k_1,k_3,\cdots k_{N-1},k_{N-2},\cdots k_4,k_2]\Delta f_{\min}
    \end{align*}
    If $k_{N-1}<M+2-N$, then there exists at least one $k_{m}\geq 2,m\neq N-1$. Without loss of generality, $N$ and $m$ are assumed to be even and odd). Then

    \begin{gather*}
    \Delta {\vec{f}}=\left[k_1,k_3,\cdots\underline {k_m }, \cdots \underline {k_{N-1}},\cdots k_4,k_2\right]^T \Delta f_{\min }
    \end{gather*}
    Define
    \begin{gather*}
    \Delta \vec{f^{'}}=\left[k_1,k_3,\cdots\underline{k_m-1},\cdots\underline{k_{N-1}+1},\cdots k_4,k_2\right]^T\Delta f_{\min}
    \end{gather*}
    Since
    \begin{align*}
    \Delta \vec {f}^T \vec{\Gamma }^T \left( N{\vec{I}}_{N-1}-\vec{uu}^T  \right){\vec{\Gamma }}\Delta {\vec{f}}
        = \mathrm{tr}\left( \vec{PQ}\right)
        = \sum_{i = 1}^{N - 1} \sum_{j = 1}^{N - 1} \vec{P}\left({i,j}\right) \vec{Q}\left({i,j}\right)
    \end{align*}
    Where $\vec{P}=\Delta{\vec{f}}\Delta {\vec{f}}^T$,
    $\vec{Q}$ is defined in (\ref{equa:Qij}).

   Let $\Delta{\vec{u}}=\left[0,0,\cdots\underbrace{-1}_{(m+1)/2},\cdots \underbrace 1_{N/2},0,\cdots0\right]^T\Delta f_{\min}$, then
   \begin{align*}
     \vec{P^{'}}&=\Delta {\vec{f^{'}}}\Delta {\vec{f^{'}}}^T  = \left(\Delta \vec{f}+\Delta \vec{u} \right)
                     \left(\Delta \vec{f}+\Delta \vec{u} \right)^T\\
     &=\vec{P}+\Delta \vec{u}\Delta \vec{f}^T +\Delta \vec{f}\Delta \vec{u}^T+\Delta \vec{u}\Delta \vec{u}^T
   \end{align*}

   Note that $\vec{Q}(i,i) + \vec{Q}(j,j) > 2\vec {Q}(i,j)$, then
   \begin{align*}
     {\mathrm{tr}}\left(\Delta \vec{u}\Delta \vec{u}^T \vec{Q}\right)&=\vec{Q}\left(\frac{m+1}{2},\frac{m + 1}
      {2} \right)+\vec{Q}\left(\frac{N}{2},\frac{N}{2}\right)-2\vec{Q}\left(\frac{m+1}{2},\frac{N}{2}\right)\\
      &>0
   \end{align*}

   Let $\vec U = \Delta \vec{u}\Delta \vec{f}^T $, since $\mathrm{Q}$ is symmetric, we have
   \begin{align*}
    \mathrm{tr}\left(\Delta\vec{f}\Delta \vec{u}^T\vec{Q}\right)&=\mathrm{tr}\left(\Delta \vec{u}\Delta \vec{f}^T
                         \vec{Q}\right)\\
        &\hspace{-0pt}= \sum_{i=(m+1)/2,N/2} \sum_{j=1}^{N-1} \vec{U}\left(i,j\right)\vec{Q}\left(i,j \right) \\
        &\hspace{-0pt}= \sum_{j=1}^{N-1} \vec{U}\left( \frac{m + 1}{2},j \right)\vec{Q}\left( \frac{m + 1}{2},j \right)
                         +\sum_{j=1}^{N-1} \vec{U}\left( N/2,j \right)\vec{Q}\left( N/2,j\right)\\
        &\hspace{-0pt}=\sum_{j=1}^{N-1}\Delta \vec{f}\left(j \right)\left[ \vec{Q}(N/2,j)-\vec{Q}\left(\frac{m+1}{2},j
                        \right)\right]\\
        &\hspace{-0pt}=\sum_{j=1}^{N/2-1}\Delta \vec{f}(j)\left[\vec{Q}\left(N/2,j\right)-\vec{Q}\left(\frac{m+1}{2},j
                        \right)\right]
                        +\sum_{j=N/2}^{N-1}\Delta\vec{f}(j)\left[\vec{Q}\left(N/2,j\right)-\vec{Q}\left(\frac{m+1}{2},j
                        \right) \right] \\
        &\hspace{-0pt}=\sum_{j=1}^{N/2-1} \Delta \vec{f}(j)\left[\frac{N}{2}j -\vec{Q}\left(\frac{m + 1}{2},j\right)
                         \right]
                        +\sum_{j=N/2}^{N-1} \Delta \vec{f}(j)\left[ \frac{N}{2}(N-j)-\frac{m + 1}{2}(N-j) \right]
   \end{align*}
  Let $k=N-1-j$, then
  \begin{align*}
      &= \sum_{j=1}^{N/2-1} \Delta \vec{f}(j)\left[\frac{N}{2}j - \vec{Q}\left(\frac{m+1}{2},j\right)\right]
         +\Delta \vec{f}(N-1)\left[\frac{N}{2} - \frac{m+1}{2} \right]\\
          &\hspace{15pt}+\sum_{k=1}^{N/2-1} \Delta \vec{f}(N-1-k)\left[\frac{N}{2}(k+1)-\frac{m + 1}{2}(k+1) \right]\\
      &> \sum_{j=1}^{N/2-1} \Delta \vec{f}(j)\left[\frac{N}{2}j- \vec{Q}\left(\frac{m + 1}{2},j \right)\right]
          \hspace{0pt}+\sum_{j=1}^{N/2-1} \Delta \vec{f}(N-1-j)\left[\frac{N}{2}(j+1) - \frac{m + 1}{2}( j+1) \right]
  \end{align*}

  By Theorem \ref{theo:optimal_permutation}, $\Delta \vec{f}(j)<\Delta \vec{f}(N-1-j)$, therefore
  \begin{align*}
       \mathrm{tr}( \Delta \vec{f}\Delta \vec{u}^T \vec{Q})&
       >\sum_{j=1}^{N/2-1}\Delta\vec{f}(j)\Bigg[\frac{N}{2}(2j+1)-\vec{Q}\left(\frac{m+1}{2},j\right)-\frac{m+1}{2}(j+1) \Bigg]\\
       &\hspace{-0pt}>\sum_{j=1}^{N/2-1}\Delta\vec{f}(j)\left[jN-\vec{Q}\left(\frac{m+1}{2},j \right)-\frac{m + 1}{2}j
                   \right]\\
       &\hspace{-0pt}=\sum_{j=1}^{(m+1)/2}\Delta\vec{f}(j)\underbrace{\left[jN-j\left(N-\frac{m+1}{2}\right)-\frac{m+1}{2}j
                  \right]}_{=0}\\
       &\hspace{5pt}+\sum_{j=(m+2)/2}^{N/2-1}\Delta\vec{f}(j)\underbrace{\left[jN-\frac{m+1}{2}(N-j)-\frac{m+1}{2}j
                  \right]}_{>0}\\
       &\hspace{-0pt}>0
  \end{align*}

  It follows that
  \begin{gather}
  \label{equa:exchange}
    \Delta \vec{f}^T\vec{Q}\Delta \vec{f}< \Delta \vec{f^{'}}^T\vec{Q}\Delta \vec{f^{'}}
  \end{gather}

  This implies that the sequence $\Delta \vec f^{'}$ is superior to $\Delta \vec f$. Therefore, the maximum value $k_{N-1}$ of the optimal frequency spacing $\Delta \vec{f}^{\ast}$, normalized by $\Delta f_{\min}$, must satisfy $k_{N-1}=M+2-N$. Otherwise, for any $\Delta \vec f$ with $k_{N-1}<M+2-N$, we can always find another frequency spacing better than it on the basis of (\ref{equa:exchange}).

  The key idea of the proof is that all the $\Delta\vec f$ may be classified into different categories (categorized into groups). Those $\Delta \vec f$ composed of the same set of frequency spacing belong to the same class, which differ only in the permutation. For each class, the optimal one is easily obtained by Theorem \ref{theo:optimal_permutation}. The global optimal across all the local optimal takes the form:
      \begin{gather*}
          \Delta \vec{f}^{\ast}=[1,1,\cdots M+2-N,\cdots 1,1]^T \Delta f_{\min}
      \end{gather*}

  It is pointed that when the frequencies are densely distributed on both ends of the band, those frequencies can be approximated as two frequencies on both ends and the corresponding UMR approaches $c/B$. Thus, the condition of $q \in [-c/2B,c/2B]$ is imposed on the Theorem to exclude ambiguity.




\ifCLASSOPTIONcaptionsoff
  \newpage
\fi



\bibliographystyle{IEEEtran}









\end{document}